\documentclass[11pt]{article}
\usepackage{times,fullpage,amssymb,amsmath,amsthm,hyperref,xcolor,multirow,makecell,authblk}

\newtheorem{lemma}{Lemma}[section]
\newtheorem{theorem}[lemma]{Theorem}
\newtheorem{claim}[lemma]{Claim}
\newtheorem{problem}{Problem}
\newtheorem{definition}{Definition}

\title{Improved Strongly Polynomial Work-Span Tradeoffs for \\ Directed Single Source Shortest Paths}

\begin{document}
    \author[1]{Isaac M. Hair}
    \author[2]{George Z. Li}
	\author[2]{Jason Li}
	\author[2]{Junkai Zhang}
	\affil[1]{UCSB, UCLA    \texttt{isaacmhair@gmail.com}}
	\affil[2]{Carnegie Mellon University $\{\texttt{gzli}, \texttt{jmli},\texttt{junkaizh}\}$\texttt{@cs.cmu.edu}}
    \date{}
\maketitle

\begin{abstract}
    We revisit the single-source shortest paths (SSSP) problem on directed graphs with nonnegative real weights and give a deterministic parallel algorithm with $O(n^{1+o(1)}t^2 + m^{1+o(1)})$ work and $\tilde{O}(n/t)$ span, for any $t \in [1, n]$. This matches (up to subpolynomial factors) the tradeoff due to [Shi and Spencer '99] for \emph{undirected} graphs with nonnegative real weights.
\end{abstract}

\section{Introduction}

Consider the single-source shortest paths problem on a directed graph with non-negative \emph{real} edge weights, also known as the non-negative strongly polynomial setting. We study this problem under the \emph{work-span} model in parallel (PRAM) computing, where the work of a parallel algorithm is its sequential running time, and the span is the longest chain of dependencies required by the algorithm. There are two natural parallel algorithms in this setting: Dijkstra's algorithm, which can be implemented in $\tilde O(m)$ work and $\tilde O(n)$ span,\footnote{$\tilde O(\cdot)$ hides factors poly-logarithmic in $n$.} and the repeated squaring algorithm that computes $O(\log n)$ many min-plus products, which takes $O(n^3)$ work and $\text{polylog}(n)$ span. A long line of work has obtained work-span tradeoffs between these two extremes; see Table~\ref{fig:past-work} for a detailed list.

In this work, we obtain the following work-span tradeoff.

\begin{theorem}[Informal version of Theorem \ref{thm:final}]
There is a deterministic parallel algorithm that computes single-source shortest paths on non-negative real-weighted directed graphs with $O(n^{1+o(1)}t^2 + m^{1+o(1)})$ work and $\tilde{O}(n/t)$ span, for any $t \in [1,n]$, where $n$ and $m$ are the number of vertices and edges in the graph respectively.
\end{theorem}

Notably, our work-span tradeoff is the first to simultaneously match Dijkstra's algorithm ($t=1$) and repeated squaring ($t=n$). We also match the tradeoff of Shi and Spencer for \emph{undirected} single-source shortest paths~\cite{shi1999time}, and the tradeoff of Spencer for directed single-source \emph{reachability}~\cite{spencer1997time}, which was subsequently generalized to strongly connected components with the same tradeoff~\cite{karczmarz2025deterministic}.

\begin{table}
  \centering
  \def\arraystretch{1.4}%
  \begin{tabular}{|c|c|c|c|}
  \hline
  \textbf{Algorithm} & \textbf{Work} & \textbf{Depth} & \textbf{Restrictions} \\ \hline
    parallel Dijkstra \cite{paige1985parallel} & $\tilde O(m)$ & $\tilde O(n)$ & --- \\ \hline
    \makecell{repeated min-plus product~\cite{williams2014faster}} & $n^{3-o(1)}$ & $\text{polylog}(n)$ & --- \\ \hline
    \makecell{Spencer~\cite{spencer1997time}} & $\tilde O(n^2t)$ & $\tilde O(n/t)$ & --- \\ \hline
    \makecell{Bringmann, Hansen, and Krinninger~\cite{bringmann2017improved}} & $\tilde O(mn+t^3)$ & $\tilde O(n/t)$ & --- \\ \hline
    \multirow{2}{*}{Shi and Spencer \cite{shi1999time}} & $\tilde O(m + nt^2)$ & $\tilde O(n/t)$ & \multirow{2}{*}{undirected} \\ \cline{2-3}
   & $\tilde O(mn/t + t^3)$ & $\tilde O(n/t)$ & \\ \hline
    \multirow{2}{*}{Karczmarz, Nadara, and Sokołowski~\cite{karczmarz2026strongly}} & $\tilde O(m^{5/3}t^2+m^{3/2}t^{7/2})$ & $\tilde O(m/t)$ & $t \le m^{1/2}$ \\ \cline{2-4}
     & $\tilde O(m + n^{9/5}t^{17/5})$ & $\tilde O(n/t)$ & $t \le n^{1/17}$ \\ \hline
    \makecell{\textbf{This work}} & $O(m^{1+o(1)} + n^{1+o(1)}t^2)$ & $\tilde{O}(n/t)$ & --- \\ \hline
\end{tabular}
  \caption{History of parallel algorithms for non-negative strongly polynomial single-source shortest paths, originally from~\cite{karczmarz2026strongly}. The tradeoff parameter $t$ has default range $[1,n]$.}\label{fig:past-work}
\end{table}

\subsection{Techniques}

At a very high level, our algorithm is a twice-hierarchical adaptation of the algorithm of Karczmarz, Nadara, and Sokołowski~\cite{karczmarz2026strongly}: we construct a hierarchy of the \emph{heavy} vertices from~\cite{karczmarz2026strongly}, as well as a separate \emph{hop-bounded} hierarchy of \emph{near-lists}, an object introduced by Spencer~\cite{spencer1997time} and also used by~\cite{karczmarz2026strongly}.

In more detail, the algorithm tries to simulate $t$ consecutive steps of Dijkstra's algorithm at a time, similar to previous algorithms~\cite{spencer1997time,shi1999time,karczmarz2026strongly}. This requires finding the $t$ unprocessed vertices that are closest to the source, and then processing them. Suppose that optimistically, each vertex $v$ can maintain a \emph{near-list} of the $t$ unprocessed vertices closest to $v$, i.e., the unprocessed vertices $u$ with smallest distance from $v$ to $u$. Then, the next $t$ vertices in Dijkstra's algorithm are contained in the near-lists of the processed vertices. The algorithm can retrieve these $t$ vertices by maintaining a parallel priority queue of the near-list vertices, keyed by their estimated distances, and then calling a batched extraction of the $t$ minimum keys.

\paragraph{Heavy vertices.}

There are two issues with this approach. First, since the near-lists are restricted to unprocessed vertices, they may change as new vertices get processed. More precisely, when a vertex gets processed, all near-lists containing that vertex must be updated to find a new unprocessed vertex. A single unprocessed vertex may belong to many near-lists, in which case processing that vertex forces many near-lists to be updated. Karczmarz, Nadara, and Sokołowski~handle this issue by declaring such vertices to be \emph{heavy}. Heavy vertices are treated separately, and near-lists are then re-defined to be the $t$ closest unprocessed vertices on the induced subgraph without heavy vertices. These near-lists may induce new heavy vertices, so this process must be iterated, but the number of heavy vertices remains relatively small.

Our algorithm follows the same approach, but we do not stop at one level of heavy vertices. To build the second level of heavy vertices, we restrict our attention to the near-lists of the (first level of) heavy vertices, and declare a vertex to be heavy at the second level if it belongs to many of them. We ensure that each subsequent level has at most half the heavy vertices of the previous level, so $O(\log n)$ levels suffice for the hierarchy.

\paragraph{Hop-bounded searches.}

The second issue is that computing near-lists of size $t$ may require computing $t$-hop shortest paths. For example, if the input graph is a directed path, then the near-list of each unprocessed vertex is the $t$ unprocessed vertices after it on the path. Naively, (re-)computing a near-list using $t$-hop shortest paths requires $O(t)$ span, which is useless because spending $O(t)$ span to simulate $t$ steps of Dijkstra results in $O(n)$ span overall.

We handle this issue by introducing another hierarchy based on hop-length, i.e., the number of edges in the shortest path between two vertices. On a level with hop-bound $h$, the near-lists maintain the $t$ closest unprocessed vertices (on the induced subgraph without heavy vertices), where distance is measured by the shortest path with at most $h$ edges. The first level $h=1$ is trivial, since the near-lists are simply the $t$ closest neighbors. Each subsequent level has a geometrically larger hop-bound, and near-lists at this level are re-computed by using the near-lists at the previous level. For technical reasons, we require that the hop-bounds increase by factor $2^{\sqrt{\log n}}$ per level, which explains the $n^{o(1)}$ factors in the work.

\subsection{Related Work}

Parallel single-source shortest paths has also been studied in more restricted regimes. In the case of non-negative \emph{integral} weights, which allows for scaling-based algorithms, the state of the art algorithms have work $\tilde O(m\log W)$ and span $n^{1/2+o(1)}\log W$, where $W$ is the maximum weight~\cite{cao2023parallel,rozhovn2023parallel}. (There is also a recent improvement for non-sparse graphs~\cite{ashvinkumar2026parallel}.) In fact, these bounds are known even if the graph has \emph{negative} integral weights at least $-W$~\cite{ashvinkumar2024parallel,fischer2025simple}, a setting incomparable to our non-negative real-weighted graphs.

Another fruitful line of work considers \emph{$(1+\epsilon)$-approximate} single-source shortest paths in undirected graphs~\cite{cohen2000polylog,elkin2019hopsets}, where truly work-efficient parallel algorithms are known with $\tilde O(m)$ work and $\text{polylog}(n)$ span~\cite{andoni2020parallel,li2020faster,rozhovn2022undirected}.

\section{Preliminaries}
Let $G_0 = (V_0, E_0, w)$ be the input directed graph, where $\vert V_0\vert = n$, $\vert E_0\vert = m$, and $w : E_0 \rightarrow \mathbb{R}_{\geq 0}$ is a weight function. Assume without loss of generality that $m \geq n$, for example by deleting isolated vertices. Also assume that every vertex in $G_0$ has in-degree and out-degree at most $O(m/n)$. This can be enforced, for example, by replacing each vertex whose in/out-degree is $\omega(m/n)$ by a tree of degree $O(m/n)$. The total number of vertices and edges then increases by at most a constant factor.

Now let $G = (V, E, w)$ be any graph (throughout the paper, all graphs will be nonnegative real weighted directed graphs). Let $V(G)$ denote the vertex set of the graph. For a path $\pi$ in $G$, let $w(\pi)$ be the total weight of the path. We refer to a path using at most $h$ edges as an \emph{$h$-hop path}. For a subset $V' \subseteq V$, denote by $G[V']$ the induced subgraph of $G$ on the vertices in $V'$. Given two vertices $v, u \in V$, we use $d_{G}(v, u)$ to denote the shortest path distance from $v$ to $u$ in $G$ (if such a path exists; otherwise we set $d_{G}(v, u) = +\infty$). Similarly, we use $d_{G}^h(v, u)$ to denote the length of the shortest $h$-hop path from $v$ to $u$ in $G$ (or set $d_{G}^h(v, u) = +\infty$ if no such path exists). When a distance in an induced subgraph is written with an endpoint outside the vertex set of that induced subgraph, we interpret the corresponding distance as $+\infty$.

We use $\tilde{O}(\cdot)$ to suppress polylogarithmic factors, i.e.\ $\tilde{O}(T(n)) = O(T(n)\log^{O(1)} n)$. Denote by $[a, b]$ the set $\{a, a+1, \ldots, b\}$, and denote by $[a]$ the set $\{1, 2, \ldots, a\}$. We will make extensive use of parallel priority queues and parallel binary search trees (see, e.g., \cite{brodal1998parallel, blelloch2016just}). These data structures allow us to efficiently perform the standard operations supported by priority queues and binary search trees, except that operations of the same type can be grouped into batches of an arbitrary size $b$ and performed simultaneously in $O(b\log^{O(1)}n)$ work and $O(\log^{O(1)}n)$ span (where $n^{O(1)}$ is an upper bound on the total size of the data structure). We also use sorting networks (see, e.g., \cite{batcher1968sorting}), which allow us to sort $b$ real values in $O(b\log^{O(1)}b)$ work and $O(\log^{O(1)}b)$ span. All of these primitives have deterministic implementations.

Our SSSP algorithm will maintain copies of the data structure below, for different choices of $h$ and $k$.

\begin{definition}[$(h, k)$-Layered Near-List (LNL) Data Structure.] \label{def:LNL} \normalfont
Let $G = (V, E, w)$ be an induced subgraph of the original graph $G_0$, and let $h \in [1, n]$ and $k \in [1, n]$ be parameters\footnote{Recall that $n$ is the number of vertices in $G_0$, not the subgraph $G$.}. Set $\ell$ to be the smallest positive integer such that $n/2^\ell \leq k$. An $(h, k)$-LNL data structure for $G$ consists of the following.
\begin{enumerate}
    \item The graph $G$, stored such that every vertex $v \in V$ has a parallel priority queue containing its out-neighbors, ordered by the weights of the edges from $v$ to each out-neighbor. (The graph will be modified over time by deleting subsets of vertices, and the priority queues will be updated in response.)
    \item Subsets $H_0, \ldots, H_{\ell+1} \subseteq V$, where $H_0 = V$, $\vert H_i\vert \leq n/2^i$ for all $i \in [1, \ell]$, and $H_{\ell+1} = \emptyset$. We refer to these as \emph{heavy vertex subsets.}
    \item For all $i \in [0, \ell]$ and for all $v \in H_i$, a near-list $\text{NL}_i^h(v) \subseteq V$ satisfying the following invariants.
    \begin{enumerate}
        \item \textbf{Size Bound.} $\vert \text{NL}_i^h(v)\vert \leq 2k$.
        \item \textbf{Completeness.} If $\vert \text{NL}_i^h(v)\vert \leq k$, then there are at most $k$ vertices reachable from $v$ via an $h$-hop path in $G[V\backslash H_{i+1}]$,\footnote{If $v \in H_{i+1}$, then there are no vertices reachable from $v$ in $G[V\backslash H_{i+1}]$.} and all such vertices appear in $\text{NL}_i^h(v)$.
        \item \textbf{Distance Estimate.} Every vertex $u \in \text{NL}_i^h(v)$ has an accompanying finite distance $\tilde{d}_i(v, u)$ that satisfies
        \[d_{G_0}(v, u) \leq \tilde{d}_i(v, u) \leq d_{G[V\backslash H_{i+1}]}^h(v, u).\]
        \item \textbf{Ordering.} For all $\hat{u} \in V \backslash (H_{i+1} \cup \text{NL}_i^h(v))$ and for all $u \in \text{NL}_i^h(v)$, we have
        \[d_{G[V\backslash H_{i+1}]}^h(v, \hat{u}) \geq \tilde{d}_i(v, u).\]
    \end{enumerate}
    \textbf{Congestion.} For each $i \in [0,\ell]$, we also require that for all $u \in V$, there are at most $10k$ vertices $v \in H_i$ such that $u \in \text{NL}_i^h(v)$. (For $i = \ell$, there are at most $k$ near-lists, so this requirement is automatically satisfied.)
\end{enumerate}
\end{definition}

\noindent
Below we formalize the problem of instantiating an LNL data structure.
\begin{problem}\label{prob:init}
    The input is the (original) graph $G_0 = (V_0, E_0, w)$, along with two integers $h \in [1, n]$ and $k \in [1, n]$. The task is to construct an $(h, k)$-LNL data structure for $G_0$. We require that, upon initialization, the data structure satisfies the following additional properties:
    \begin{enumerate}
        \item For all $i \in [1, \ell]$, we have\footnote{This means that, for example, we might have $H_\ell = \emptyset$.} $\vert H_i\vert \leq n/4^i$.
        \item For all $i \in [0, \ell]$ and for all $v \in H_i$, if $\vert \text{NL}_i^h(v)\vert < 2k$, then there are strictly fewer than $2k$ vertices reachable from $v$ via an $h$-hop path in $G[V\backslash H_{i+1}]$, and all such vertices appear in $\text{NL}_i^h(v)$.
    \end{enumerate}
\end{problem}

\paragraph{Updating LNL data structures.}
As our SSSP algorithm progresses, we will update each $(h, k)$-LNL data structure using the following operations:
\begin{enumerate}
    \item \textbf{Batched Deletion.} Given a subset $Z \subseteq V$ such that $\vert Z\vert \leq k$, delete the vertices in $Z$ from the graph $G$ stored in the $(h, k)$-LNL data structure. Update all components of the data structure to maintain the required invariants with respect to this new graph.
    \item \textbf{Place Marker.} Mark the current state of all components in the $(h, k)$-LNL data structure (including the current version of the graph $G$).
    \item \textbf{Return to Marker.} Revert all components of the data structure to the state indicated by the most recent marker (and then remove the marker).
\end{enumerate}
We refer to these three operations as \emph{update queries}. They allow us to maintain a collection of near-lists with respect to the unsettled vertices of $G_0$.

In the analysis, it will be useful to reason about the combinatorial properties of an update sequence. Define the \emph{transcript} of an update sequence of length $r$ to be the length $r+1$ vector $\mathbf{t} \in \mathbb{Z}^{r+1}$ that starts with the initial number of vertices in $G$, and then records the total number of vertices in $G$ after each update. Now define the \emph{transcript width} of such an update sequence as
\[\max_{c \in [0,n]} \vert \{z \in [r] : \mathbf{t}_z > c \geq \mathbf{t}_{z+1}\}\vert.\]
In other words, the transcript width records the maximum, over all choices of vertex count $c$, of the number of times that the graph $G$ in the $(h, k)$-LNL data structure will transition from having more than $c$ vertices to having at most $c$ vertices, over the entire sequence of updates.

\section{Using LNL Data Structures to Perform Searches}\label{sec:search}
LNL data structures are useful for parallel SSSP because they allow us to efficiently perform a hop-restricted search in $G$ from any given query vertex. In the final algorithm, this will allow us to grow the required shortest path tree approximately $k$ vertices at a time, with each step requiring only $n^{o(1)}$ span. We start by giving a supporting lemma for the searches:

\begin{lemma} \label{lem:fromStoSprime}
    Given access to an $(h, k)$-LNL data structure for an induced subgraph $G = (V, E, w)$ of the original graph $G_0$, there is a deterministic parallel algorithm with $\tilde{O}(k^2)$ work and $\tilde{O}(1)$ span that behaves as follows. It takes as input a query subset $S \subseteq V$ such that $\vert S\vert \leq k$, along with a distance estimate $\tilde{d}(s) \in \mathbb{R}_{\geq 0}$ for each $s \in S$, and it outputs a subset $S' \subseteq V$, along with a distance estimate $\tilde{d}'(s') \in \mathbb{R}_{\geq 0}$ for each $s' \in S'$, such that the following holds:
    \begin{enumerate}
        \item \textbf{Size Bound.} $\vert S'\vert \leq k$.
        \item \textbf{Completeness.} If $\vert S'\vert < k$, then strictly fewer than $k$ vertices are reachable from a vertex in $S$ via an $h$-hop path in $G$, and all such vertices appear in $S'$.
        \item \textbf{Distance Estimate.} $\min_{s \in S} (\tilde{d}(s) + d_{G_0}(s, s')) \leq \tilde{d}'(s') \leq \min_{s \in S} (\tilde{d}(s) + d_G^h(s, s'))$ for all $s' \in S'$.
        \item \textbf{Ordering.} For all $\hat{s} \in V \backslash S'$ and for all $s' \in S'$, we have $\min_{s \in S} (\tilde{d}(s) + d_G^h(s, \hat{s})) \geq \tilde{d}'(s')$.
    \end{enumerate}
\end{lemma}

\begin{proof}
    Initialize $X_0 = S$, and for each $s \in X_0$ set $\Delta_0(s) = \tilde{d}(s)$. For $i = 0, 1, \ldots, \ell$, do the following:
    \begin{enumerate}
        \item \textbf{Near-List Relaxation.} Compute
        \[X_i^{(1)} = X_i \cup \bigcup_{v \in X_i \cap H_i}\text{NL}_i^h(v).\]
        Then for each $u \in X_i^{(1)}\backslash X_i$, set
        \[\Delta_i'(u) = \min_{\substack{v \in X_i \cap H_i :\\ u \in \text{NL}_i^h(v)}} (\Delta_i(v) + \tilde{d}_i(v, u)),\]
        where $\tilde{d}_i(v,u)$ is the distance stored in the $(h, k)$-LNL data structure for $G$. Similarly, for each $u \in X_i$, set
        \[\Delta_i'(u) = \min \Big(\Delta_i(u), \min_{\substack{v \in X_i \cap H_i :\\ u \in \text{NL}_i^h(v)}} (\Delta_i(v) + \tilde{d}_i(v, u))\Big).\]
        The work is $\tilde{O}(k^2)$ and the span is $\tilde{O}(1)$, because we can just read the appropriate sets from the $(h, k)$-LNL data structure for $G$, and then take the minimums in parallel.
        \item \textbf{Truncation.} Sort all $u \in X_i^{(1)}$ by $\Delta_i'(u)$, and let $X_i^{(2)} \subseteq X_i^{(1)}$ be the set containing the $k$ vertices with smallest $\Delta_i'$ from $X_i^{(1)}$. (Ties are broken arbitrarily. If $\vert X_i^{(1)}\vert < k$, then set $X_i^{(2)} = X_i^{(1)}$.) Using a sorting network, the work is $\tilde{O}(k^2)$ and the span is $\tilde{O}(1)$.
        \item \textbf{Edge Relaxation.} For all $u \in X_i^{(2)}$, let $F(u)$ be either (i) the $k$ closest out-neighbors of $u$ in $G$, if $u$ has at least $k$ out-neighbors in $G$, or else (ii) the set of all out-neighbors of $u$ in $G$. Now compute
        \[X_i^{(3)} = X_i^{(2)} \cup \bigcup_{u \in X_i^{(2)}} F(u).\]
        For each $z \in X_i^{(3)}\backslash X_i^{(2)}$, set
        \[\Delta_{i+1}(z) = \min_{\substack{u \in X_i^{(2)}\\ z \in F(u)}} (\Delta_i'(u) + w(u,z)),\]
        and for each $z \in X_i^{(2)}$, set
        \[\Delta_{i+1}(z) = \min \big(\Delta_i'(z), \min_{\substack{u \in X_i^{(2)}\\ z \in F(u)}} (\Delta_i'(u) + w(u,z))\big).\]
        Because the $(h, k)$-LNL data structure for $G$ already has each vertex stored along with a parallel priority queue for its out-neighbors, the work is $\tilde{O}(k^2)$ and the span is $\tilde{O}(1)$.
        \item \textbf{Truncation.} Sort all $z \in X_i^{(3)}$ by $\Delta_{i+1}(z)$, and let $X_{i+1} \subseteq X_i^{(3)}$ be the set containing the $k$ vertices with smallest $\Delta_{i+1}$ from $X_i^{(3)}$. (As before, ties are broken arbitrarily, and if $\vert X_i^{(3)}\vert < k$, then set $X_{i+1} = X_i^{(3)}$.)
    \end{enumerate}
    Take the resulting set $X_{\ell+1}$ as the set $S'$, and return $\tilde{d}'(s') = \Delta_{\ell+1}(s')$ for each $s' \in S'$. Because there are $O(\log n)$ rounds, the total work is $\tilde{O}(k^2)$, and the total span is $\tilde{O}(1)$. All steps are deterministic.

    All that remains is to argue that $S'$ and the distance estimates $\tilde{d}'(s')$ for each $s' \in S'$ satisfy the requirements of the lemma. As shorthand, for each $u \in V$, set
    \[
        \Psi(u)=\min_{s \in S}\bigl(\tilde{d}(s)+d_{G_0}(s,u)\bigr)
        \qquad\text{and}\qquad
        \Phi(u)=\min_{s \in S}\bigl(\tilde{d}(s)+d_G^h(s,u)\bigr).
    \]
    Also define a parameter $\tilde{d}'_{\max}$ as follows: if $\vert S'\vert = k$, then set $\tilde{d}'_{\max} = \max_{s' \in S'} \tilde{d}'(s')$, and otherwise set $\tilde{d}'_{\max} = +\infty$. We start by giving a simple consequence of the definition of $\tilde{d}'_{\max}$.

    \begin{claim} \label{claim:dominance}
        For all $i \in [0, \ell]$, we have $\vert \{u \in X_i^{(1)} : \Delta_i'(u) < \tilde{d}'_{\max}\}\vert < k$, and we have $\vert \{z \in X_i^{(3)} : \Delta_{i+1}(z)  < \tilde{d}'_{\max}\}\vert < k$.
    \end{claim}

    \begin{proof}
        Suppose for contradiction that for some $i \in [0, \ell]$, either $\vert \{u \in X_i^{(1)} : \Delta_i'(u) < \tilde{d}'_{\max}\}\vert \geq k$, or $\vert \{z \in X_i^{(3)} : \Delta_{i+1}(z)  < \tilde{d}'_{\max}\}\vert \geq k$. Then every subsequent set $X_{i+1}$ constructed by the algorithm will have at least $k$ vertices whose distance estimate $\Delta_{i+1}$ is strictly less than $\tilde{d}'_{\max}$. By definition of $S'$ and $\tilde d'(s')$, it follows that $|S'|=k$ and $\max_{s' \in S'} \tilde{d}'(s') < \tilde{d}'_{\max}$, contradicting the definition of $\tilde{d}'_{\max}$.
    \end{proof}
    
    The following claim is the main ingredient in the correctness argument.

    \begin{claim} \label{claim:includey}
        Let $y \in V$ be any vertex such that, for some $s \in S$, there is an $h$-hop path in $G$ from $s$ to $y$. If $\Phi(y)<\tilde{d}'_{\max}$, then $y \in S'$ and $\tilde{d}'(y)\leq \Phi(y)$.
    \end{claim}

    \begin{proof}
        Let $\pi = (v_0 = s, v_1, \ldots, v_r=y)$ be a path using at most $h$ edges such that $s \in S$ and $\tilde{d}(s)+w(\pi)=\Phi(y)$. For each vertex $v_j$ of $\pi$, use
        \[\text{pref}(v_j)=\tilde{d}(s)+\sum_{a=0}^{j-1}w(v_a,v_{a+1})\]
        to denote the length of the subpath from $s$ to $v_j$ together with the distance estimate for $s$.

        Now set $p_0 = s$, and for each $i \in [1, \ell]$ set $p_i$ to be the first vertex weakly after\footnote{This means that we allow $p_i = p_{i-1}$.} $p_{i-1}$ along the path $\pi$ such that $p_i \in H_i$ (if no such vertex exists, then set $p_i = y$). Also set $p_{\ell+1} = y$. We will prove inductively that for all $i \in [0, \ell+1]$, we have $p_i \in X_i$, and additionally $\Delta_i(p_i) \leq \text{pref}(p_i)$.

        The base case is immediate: $p_0=s\in X_0$ and $\Delta_0(p_0)=\tilde{d}(s)=\text{pref}(p_0)$. Assume now that the invariant holds for some $i\leq \ell$. We prove it for $i+1$.

        First suppose that $p_i=y$. Then $p_{i+1}=y$. By the inductive hypothesis, $y\in X_i$, so by Step 1 of the algorithm, we have $y \in X_i^{(1)}$ with distance estimate $\Delta_i'(y) \leq \Delta_i(y)$. Again using the inductive hypothesis, we know $\Delta_i(y) \leq \text{pref}(y)$, and by the assumption in the claim, we have $\text{pref}(y) = \Phi(y) < \tilde{d}'_{\max}$. This implies $\Delta_i'(y) < \tilde{d}'_{\max}$, so by Claim \ref{claim:dominance}, there are strictly less than $k$ vertices in $X_i^{(1)}$ whose distance estimate $\Delta_i'$ is at most $\Delta_i'(y)$, and by Step 2 of the algorithm, we have $y \in X_i^{(2)}$. Using the same analysis for Steps 3 and 4 of the algorithm, we have $y \in X_i^{(3)}$ and $y \in X_{i+1}$, with distance estimate $\Delta_{i+1}(y) \leq \text{pref}(y)$.

        Now suppose that $p_i\neq y$, which means $p_i\in H_i$. If $p_i\in H_{i+1}$, then by definition of the sequence $p_0,p_1,\ldots,p_{\ell+1}$, we have $p_{i+1}=p_i$. As in the paragraph above, $p_{i+1}$ is carried through all steps at iteration $i$ with estimate $\Delta_{i+1}(p_{i+1}) \leq \text{pref}(p_{i+1})$, so the invariant holds for $i+1$.

        It remains to handle the case that $p_i\notin H_{i+1}$. Let $b$ be the first vertex after $p_i$ on $\pi$ that belongs to $H_{i+1}$, if such a vertex exists. If $b$ exists, then $p_{i+1}=b$, and we let $c$ be the vertex immediately preceding $b$ on $\pi$. If no such $b$ exists, then $p_{i+1}=y$, and we set $c=y$. In either case, the subpath of $\pi$ from $p_i$ to $c$ is contained in $G[V\backslash H_{i+1}]$ and uses at most $h$ edges.

        We claim that, after Step 2 of the algorithm at iteration $i$, we have $c \in X_i^{(2)}$ and $\Delta_i'(c) \leq \text{pref}(c)$. If $c=p_i$, this is immediate from the inductive hypothesis and Claim \ref{claim:dominance}. Otherwise, if $c\in \text{NL}_i^h(p_i)$, then Step~1 discovers $c$ with an estimate of at most
        \[
            \Delta_i'(c) \leq \Delta_i(p_i)+\tilde{d}_i(p_i,c)
            \leq \text{pref}(p_i)+d_{G[V\backslash H_{i+1}]}^h(p_i,c)
            \leq \text{pref}(c),
        \]
        and Claim \ref{claim:dominance} (combined with the observation that $\text{pref}(c) \leq \Phi(y)$ and the assumption that $\Phi(y) < \tilde{d}'_{\max}$) implies that $c$ is retained in Step 2. Finally, suppose that $c\notin \text{NL}_i^h(p_i)$. Since $c$ is reachable from $p_i$ via an $h$-hop path in $G[V\backslash H_{i+1}]$, the completeness property of $\text{NL}_i^h(p_i)$ implies $|\text{NL}_i^h(p_i)|>k$. Moreover, by the ordering property, every $u\in \text{NL}_i^h(p_i)$ satisfies
        \[
            \tilde{d}_i(p_i,u)
            \leq d_{G[V\backslash H_{i+1}]}^h(p_i,c)
            \leq \text{pref}(c)-\text{pref}(p_i).
        \]
        Thus the near-list relaxation in Step 1 generates more than $k$ vertices with estimates at most $\Delta_i(p_i)+(\text{pref}(c)-\text{pref}(p_i)) \leq \text{pref}(c)\leq \Phi(y)$. Because we assumed $\Phi(y) < \tilde{d}'_{\max}$, this contradicts Claim \ref{claim:dominance}. This proves the claim about $c$.

        We return to proving the inductive statement, namely $p_{i+1} \in X_{i+1}$ and $\Delta_{i+1}(p_{i+1}) \leq \text{pref}(p_{i+1})$. If the vertex $b$ does not exist, then $p_{i+1}=y=c$, and the fact that $y$ is carried through Steps 3 and 4 again follows from Claim \ref{claim:dominance}. Thus the invariant holds for $i+1$ in this case. Otherwise, $(c,b)$ is an edge of $\pi$. If $b\in F(c)$, then the edge relaxation generates $b=p_{i+1}$ with estimate at most
        \[
            \Delta_{i+1}(p_{i+1}) \leq \text{pref}(c)+w(c,b)=\text{pref}(p_{i+1}),
        \]
        and Claim \ref{claim:dominance} implies that $p_{i+1} \in X_{i+1}$. If $b\notin F(c)$, then $F(c)$ consists of $k$ out-neighbors of $c$, each joined to $c$ by an edge of weight at most $w(c,b)$. The edge-relaxation step would therefore produce $k$ vertices with estimates at most $\text{pref}(b)\leq \Phi(y)$, contradicting Claim \ref{claim:dominance}.

        By induction, the invariant holds for every $i\in[0,\ell+1]$. Taking $i=\ell+1$ gives $y\in X_{\ell+1}=S'$ and $\tilde{d}'(y)=\Delta_{\ell+1}(y)\leq \Phi(y)$, as required.
    \end{proof}

    Now we verify the required properties in the lemma statement. The size bound $|S'|\leq k$ follows by construction. If $|S'|<k$, then $\tilde{d}'_{\max}=+\infty$ by definition, so Claim~\ref{claim:includey} implies that every vertex reachable from some $s\in S$ via an $h$-hop path in $G$ is contained in $S'$. Since $|S'|<k$, there are strictly fewer than $k$ such vertices. This proves the completeness property.

    For the distance estimates, we first prove the lower bound. We claim that every estimate assigned by the algorithm to a vertex $u$ is at least $\Psi(u)$. Initially this holds because $\Psi(s)\leq \tilde{d}(s)$ for all $s\in S$. Consider now a relaxation that assigns an estimate to $u$, and let $\eta(v)$ be the estimate of the vertex $v$ used in this relaxation. If this is a near-list relaxation, then by induction on the sequence of relaxations and by the lower bound on stored near-list distances,
    \[
        \eta(v)+\tilde{d}_i(v,u)
        \geq \Psi(v)+d_{G_0}(v,u)
        \geq \Psi(u).
    \]
    If this is an edge relaxation along $(v,u)$, then similarly
    \[
        \eta(v)+w(v,u)
        \geq \Psi(v)+d_G(v,u)
        \geq \Psi(u).
    \]
    Taking minima and deleting vertices clearly preserves the invariant, and hence $\tilde{d}'(s')\geq \Psi(s')$ for all $s'\in S'$.

    It remains to prove the upper bound. Suppose, for contradiction, that some $s'\in S'$ satisfies $\tilde{d}'(s')>\Phi(s')$. Since $\tilde{d}'(s')\leq \tilde{d}'_{\max}$ if $|S'|=k$, and $\tilde{d}'_{\max}=+\infty$ otherwise, Claim~\ref{claim:includey} applies to $s'$ and gives $\tilde{d}'(s')\leq \Phi(s')$, a contradiction. Thus $\tilde{d}'(s')\leq \Phi(s')$ for every $s'\in S'$, completing Item 3.

    Finally, suppose for contradiction that there exist $\hat{s}\in V\backslash S'$ and $s'\in S'$ such that $\Phi(\hat{s})<\tilde{d}'(s')$. Then $\Phi(\hat{s})<\tilde{d}'_{\max}$, so Claim~\ref{claim:includey} implies $\hat{s}\in S'$, a contradiction. Therefore, for all $\hat{s}\in V\backslash S'$ and all $s'\in S'$,
    \[
        \min_{s\in S}\bigl(\tilde{d}(s)+d_G^h(s,\hat{s})\bigr)
        =\Phi(\hat{s})
        \geq \tilde{d}'(s'),
    \]
    which is the ordering property.
\end{proof}

We are now ready to prove the main search lemma.

\begin{lemma}\label{lem:usetosearch}
    Given access to an $(h, k)$-LNL data structure for an induced subgraph $G = (V, E, w)$ of the original graph $G_0$, along with an integer $q \in [1, n]$, there is a deterministic parallel algorithm with $\tilde{O}(k^2q)$ work and $\tilde{O}(q)$ span that behaves as follows. It takes as input a query vertex $x \in V$, and it outputs a subset $Y \subseteq V$, along with a distance estimate $\tilde{d}(x, y)$ for each $y \in Y$, such that the following holds:
    \begin{enumerate}
        \item \textbf{Size Bound.} $\vert Y\vert \leq k$.
        \item \textbf{Completeness.} If $\vert Y\vert < k$, then strictly fewer than $k$ vertices are reachable from $x$ via a $qh$-hop path in $G$, and every such vertex appears in $Y$.
        \item \textbf{Distance Estimate.} $d_{G_0}(x, y) \leq \tilde{d}(x, y) \leq d_G^{qh}(x, y)$ for all $y \in Y$.
        \item \textbf{Ordering.} For all $\hat{y} \in V \backslash Y$ and for all $y \in Y$, we have $d_G^{qh}(x, \hat{y}) \geq \tilde{d}(x,y)$.
    \end{enumerate}
\end{lemma}

\begin{proof}
    Start by setting $S_0=\{x\}$ and $\Delta_0(x)=0$. For $i=0,\ldots,q-1$, apply the algorithm in Lemma~\ref{lem:fromStoSprime} with $S=S_i$ and with initial estimates $\tilde{d}(s)=\Delta_i(s)$ for all $s\in S_i$. Let $S_{i+1}$ be the returned set, and let $\Delta_{i+1}(u)$ be the returned estimate of each $u\in S_{i+1}$. At the end, take $Y=S_q$ and return $\tilde{d}(x,y)=\Delta_q(y)$ for every $y\in Y$. Since the algorithm in Lemma~\ref{lem:fromStoSprime} is invoked $q$ times, the work is $\tilde{O}(k^2q)$ and the span is $\tilde{O}(q)$ (and the algorithm is still deterministic).
    
    It remains to prove correctness. We show by induction on $i$ that after each iteration, the following properties hold:
    \begin{enumerate}
        \item \textbf{Size Bound.} $|S_i|\leq k$.
        \item \textbf{Completeness.} If $|S_i|<k$, then strictly fewer than $k$ vertices are reachable from $x$ via an $ih$-hop path in $G$, and all of them appear in $S_i$.
        \item \textbf{Distance Estimate.} $d_{G_0}(x,u)\leq \Delta_i(u)\leq d_G^{ih}(x, u)$ for all $u\in S_i$.
        \item \textbf{Ordering.} For every $\hat{u}\in V\backslash S_i$ and every $u\in S_i$, $d_G^{ih}(x, \hat{u})\geq \Delta_i(u)$.
    \end{enumerate}
    These are immediate for $i=0$.

    Suppose the properties hold for some $i<q$. In the next application of Lemma~\ref{lem:fromStoSprime}, put
    \[
        \Gamma_i(u)=\min_{s\in S_i}\bigl(\Delta_i(s)+d_G^h(s,u)\bigr).
    \]
    If $|S_{i+1}|=k$, let $\Delta_{i+1,\max}$ be the largest distance estimate in $S_{i+1}$, and otherwise put $\Delta_{i+1,\max}=+\infty$. We first prove the following implication:
    
    \begin{claim}\label{claim:imply1}
        Let $u \in V$ be any vertex such that $x$ can reach $u$ via an $(i+1)h$-hop path in $G$. If $d_G^{(i+1)h}(x,u)< \Delta_{i+1,\max}$, then $u \in S_{i+1}$ and $\Delta_{i+1}(u) \leq d_G^{(i+1)h}(x,u)$.
    \end{claim}

    \begin{proof}
        Fix a vertex $u$ such that $d_G^{(i+1)h}(x,u)<\Delta_{i+1,\max}$, and let $P$ be a shortest path from $x$ to $u$ using at most $(i+1)h$ edges. Choose a vertex $v$ on $P$ so that the prefix from $x$ to $v$ uses at most $ih$ edges and the suffix from $v$ to $u$ uses at most $h$ edges. If $\alpha$ and $\beta$ are the respective lengths of this prefix and suffix, then $\alpha+\beta=d_G^{(i+1)h}(x,u)$, and $d_G^{ih}(x, v)\leq \alpha$, and $d_G^h(v,u)\leq \beta$.
    
        If $v\in S_i$, then the inductive upper bound gives
        \[
            \Gamma_i(u)
            \leq \Delta_i(v)+d_G^h(v,u)
            \leq d_G^{ih}(x, v)+d_G^h(v,u)
            \leq d_G^{(i+1)h}(x,u).
        \]
        We claim that the alternative $v\notin S_i$ is not possible. If $|S_i|<k$, this contradicts the inductive completeness property, since $d_G^{ih}(x, v)<+\infty$. Now assume $|S_i|=k$. By the inductive ordering property and our assumption that $v\notin S_i$, every $s\in S_i$ satisfies
        \[
            \Delta_i(s)\leq d_G^{ih}(x, v)\leq d_G^{(i+1)h}(x,u).
        \]
        Since each $s\in S_i$ is reachable from itself using zero edges, we have $\Gamma_i(s)\leq d_G^{(i+1)h}(x,u)$ for all $s\in S_i$. The algorithm in Lemma~\ref{lem:fromStoSprime} therefore cannot return a set of size less than $k$, because there are already $k$ vertices reachable from $S_i$ within $h$ hops. Thus $|S_{i+1}|=k$. Using the assumption in the claim that $d_G^{(i+1)h}(x,u)<\Delta_{i+1,\max}$, we have $\Gamma_i(s) < \Delta_{i+1,\max}$ for all $s \in S_i$, so the ordering property of Lemma~\ref{lem:fromStoSprime} forces every vertex of $S_i$ to belong to $S_{i+1}$. Since both sets have size $k$, we get $S_{i+1}=S_i$. But then the distance estimate of Lemma~\ref{lem:fromStoSprime} gives $\Delta_{i+1}(s)\leq \Gamma_i(s)$, so $\Delta_{i+1}(s)\leq \Gamma_i(s)\leq d_G^{(i+1)h}(x,u)$ for every $s\in S_{i+1}$, contradicting the definition of $\Delta_{i+1,\max}$. Putting everything together, we have $v\in S_i$, and hence $\Gamma_i(u)\leq d_G^{(i+1)h}(x,u)<\Delta_{i+1,\max}$.
    
        Applying Lemma~\ref{lem:fromStoSprime} once more to the vertex $u$, we obtain $u\in S_{i+1}$: if $|S_{i+1}|<k$, this follows from completeness, while if $|S_{i+1}|=k$, the ordering property rules out $u\notin S_{i+1}$ because $\Gamma_i(u)<\Delta_{i+1,\max}$. The distance estimate then gives $\Delta_{i+1}(u)\leq \Gamma_i(u)\leq d_G^{(i+1)h}(x,u)$.
    \end{proof}

    Now we verify the four inductive properties for $i+1$. The size bound follows directly from Lemma~\ref{lem:fromStoSprime}. If $|S_{i+1}|<k$, then $\Delta_{i+1,\max}=+\infty$, so Claim~\ref{claim:imply1} implies that every vertex reachable from $x$ via an $(i+1)h$-hop path in $G$ lies in $S_{i+1}$. Hence, there are strictly fewer than $k$ such vertices.

    For the distance estimates, the lower bound follows from the distance estimates of Lemma~\ref{lem:fromStoSprime} and the inductive lower bound:
    \[
        \Delta_{i+1}(u)
        \geq \min_{s\in S_i}\bigl(\Delta_i(s)+d_{G_0}(s,u)\bigr)
        \geq \min_{s\in S_i}\bigl(d_{G_0}(x,s)+d_{G_0}(s,u)\bigr)
        \geq d_{G_0}(x,u).
    \]
    For the upper bound, suppose for contradiction that some $u\in S_{i+1}$ satisfies $\Delta_{i+1}(u)>d_G^{(i+1)h}(x,u)$. Since $d_G^{(i+1)h}(x,u)<\Delta_{i+1}(u)\leq \Delta_{i+1,\max}$, Claim~\ref{claim:imply1} gives $\Delta_{i+1}(u)\leq d_G^{(i+1)h}(x,u)$, a contradiction.

    Finally, suppose for contradiction that there are $\hat{u}\in V\backslash S_{i+1}$ and $u\in S_{i+1}$ such that $d_G^{(i+1)h}(x,\hat{u})<\Delta_{i+1}(u)$. Since $\Delta_{i+1}(u)\leq \Delta_{i+1,\max}$, Claim~\ref{claim:imply1} gives $\hat{u}\in S_{i+1}$, again a contradiction. This proves the ordering property and completes the induction.

    Taking $i=q$ gives exactly the four requirements of the lemma for $Y=S_q$.
\end{proof}

\section{Initialization}
\label{sec:init}

In this section, we show how to use an existing LNL data structure to initialize an LNL data structure with a larger hop bound. We start with the lemma below, which follows by applying the algorithm in Lemma~\ref{lem:usetosearch}.

\begin{lemma}\label{lem:performbatchedsearch}
    Given access to an $(h, 2k)$-LNL data structure for an induced subgraph $G = (V, E, w)$ of the original graph $G_0$, along with parameters $q \in [1, n]$ and $d \in [1, n]$, there is a deterministic parallel algorithm with $\tilde{O}(dk^2q)$ work and $\tilde{O}(dq/k + q)$ span that behaves as follows. It takes as input a subset $A \subseteq V$ such that $\vert A\vert = d$, makes $O(d/k + 1)$ update queries to the $(h, 2k)$-LNL data structure (of transcript width 1), and outputs a subset $B \subseteq V$ such that $\vert B\vert \leq \vert A\vert/4$, along with a near-list $\text{NL}^{qh}(v)$ for each $v \in A$, such that the following holds for each $v \in A$:
    \begin{enumerate}
        \item \textbf{Size Bound.} $\vert \text{NL}^{qh}(v)\vert \leq 2k$.
        \item \textbf{Completeness.} If $\vert \text{NL}^{qh}(v)\vert < 2k$, then there are strictly fewer than $2k$ vertices reachable from $v$ via a $qh$-hop path in $G[V\backslash B]$, and all such vertices appear in $\text{NL}^{qh}(v)$.
        \item \textbf{Distance Estimate.} Every vertex $u \in \text{NL}^{qh}(v)$ has an accompanying finite distance $\tilde{d}(v, u)$ that satisfies $d_{G_0}(v, u) \leq \tilde{d}(v, u) \leq d_{G[V\backslash B]}^{qh}(v, u)$.
        \item \textbf{Ordering.} For all $\hat{u} \in V \backslash (B \cup \text{NL}^{qh}(v))$ and for all $u \in \text{NL}^{qh}(v)$, we have $d_{G[V\backslash B]}^{qh}(v, \hat{u}) \geq \tilde{d}(v, u)$.
    \end{enumerate}
    \textbf{Congestion.} Additionally, for all $u \in V$, there are at most $10k$ vertices $v \in A$ such that $u \in \text{NL}^{qh}(v)$.
\end{lemma}

\begin{proof}
    Let $r=\lceil d/(2k)\rceil$, and arbitrarily partition $A$ into subsets $A_1,\ldots,A_r$, each of size at most $2k$. We process these subsets one at a time. During the algorithm, vertices may be deleted from the $(h,2k)$-LNL structure. Let $B_i$ denote the vertices newly deleted after processing $A_i$, define
    \[
        \overline{B}_i=\bigcup_{a=1}^i B_a,
    \]
    and let $G^{(i)}=G[V\backslash \overline{B}_i]$ be the graph stored in the $(h, 2k)$-LNL data structure after the $i$th round. Observe that $\overline{B}_0=\emptyset$ and $G^{(0)}=G$.

    We now describe the algorithm. We maintain a parallel binary search tree containing every vertex that has appeared in some constructed near-list and has not yet been deleted. We also maintain a parallel priority queue storing, for each such vertex $u$, the number of constructed near-lists that contain $u$. Now do the following for $i=1,2,\ldots,r$.
    \begin{enumerate}
        \item Start by constructing the near-lists for the vertices in $A_i$. For each $v\in A_i\cap V(G^{(i-1)})$, invoke the algorithm in Lemma~\ref{lem:usetosearch} on the current $(h, 2k)$-LNL data structure, with source vertex $v$, hop multiplier $q$, and size parameter $2k$. Let $\text{NL}^{qh}(v)\subseteq V(G^{(i-1)})$ be the returned set, and keep the returned estimates $\tilde{d}(v,u)$ for all $u\in \text{NL}^{qh}(v)$. If $v\in A_i\setminus V(G^{(i-1)})$, then set $\text{NL}^{qh}(v)=\emptyset$. All vertices can be processed independently and in parallel.
        \item For each vertex $u$ appearing in at least one newly constructed near-list, compute
        \[
            C_i(u)=\left|\{v\in A_i:u\in \text{NL}^{qh}(v)\}\right|.
        \]
        Search for all such vertices in the parallel binary search tree. If $u$ is already present, increase its priority-queue counter by $C_i(u)$. Otherwise insert $u$ into the tree and insert a new priority-queue counter initialized to $C_i(u)$. As before, all vertices can be processed in parallel.
        \item Now we identify all vertices whose current counter is at least $8k$. To do this, repeatedly extract a batch of the $2k$ largest counters from the parallel priority queue (or all remaining counters, if fewer than $2k$ remain). If all extracted counters are at least $8k$, then we add all extracted vertices to $B_i$ and continue extracting. Otherwise, we add only the extracted vertices with counter at least $8k$ to $B_i$, reinsert the remaining extracted vertices, and stop extracting. Remove the vertices of $B_i$ from the binary search tree (their priority-queue entries have already been extracted), partition $B_i$ into subsets of size at most $2k$, and issue one batched deletion query to the $(h,2k)$-LNL data structure for each subset. The $(h, 2k)$-LNL data structure now stores $G^{(i)}=G^{(i-1)}[V(G^{(i-1)})\backslash B_i]$.
    \end{enumerate}
    After all rounds, output $B=\overline{B}_r$, together with all constructed near-lists and their estimates.

    We first bound the work, span, and number of update queries. The search from a single vertex $v$ uses $\tilde{O}(k^2q)$ work and $\tilde{O}(q)$ span, because the algorithm in Lemma~\ref{lem:usetosearch} is applied with size parameter $2k$ and hop multiplier $q$. So the searches in round $i$ use $\tilde{O}(|A_i|k^2q)$ work and $\tilde{O}(q)$ span. The total size of the lists produced in round $i$ is at most $2k|A_i|=O(k^2)$, so the counter updates, priority-queue operations, and preparation of deletion batches use $\tilde{O}(k^2+|B_i|)$ work and $\tilde{O}(|B_i|/k+1)$ span. Summing over all rounds and using the bound $|B|\leq |A|/4$ proved below gives total work $\tilde{O}(|A|k^2q) = \tilde{O}(dk^2q)$ and total span
    \[
        \tilde{O}\left(\left\lceil \frac{|A|}{k}\right\rceil q+\frac{|B|}{k}+\frac{|A|}{k}+1\right)
        =\tilde{O}(|A|q/k+q) = \tilde{O}(dq/k + q).
    \]
    The only update queries made to the $(h, 2k)$-LNL data structure are the deletion queries for the sets $B_i$. Since these sets are partitioned into pieces of size at most $2k$, the number of such queries is
    \[
        \sum_{i=1}^r \left\lceil \frac{|B_i|}{2k}\right\rceil
        \leq \frac{|B|}{2k}+r
        =O(d/k+1).
    \]
    The maintained graph only loses vertices during this sequence, and therefore the transcript width is $1$.

    We next prove the size bound for $B$. Each constructed list has size at most $2k$, so the total number of counter increments over the whole algorithm is at most $2k|A|$. Vertices only enter $B$ when their counter is at least $8k$. Charging these increments to the vertices when they enter $B$ gives $8k|B|\leq 2k|A|,$ and hence $|B|\leq |A|/4$.

    We now verify the near-list guarantees. Fix $v\in A$, and suppose $v\in A_i$. If $v\in \overline{B}_{i-1}$, then $\text{NL}^{qh}(v)=\emptyset$. Since $v\notin V\backslash B$, no vertex is reachable from $v$ in the induced graph $G[V\backslash B]$. Thus all four near-list properties are immediate in this case.

    It remains to consider the case $v\notin \overline{B}_{i-1}$. The list $NL^{qh}(v)$ was constructed by applying Lemma~\ref{lem:usetosearch} in the graph $G^{(i-1)}=G[V\backslash \overline{B}_{i-1}]$, with size parameter $2k$. The size bound $|\text{NL}^{qh}(v)|\leq 2k$ follows directly. If $|\text{NL}^{qh}(v)|<2k$, then Lemma~\ref{lem:usetosearch} implies that strictly fewer than $2k$ vertices are reachable from $v$ within $qh$ hops in $G^{(i-1)}$, and all such vertices appear in the list. Since $G[V\backslash B]$ is an induced subgraph of $G^{(i-1)}$, the same statement holds for reachability in $G[V\backslash B]$. This proves completeness.

    For the distance estimates, Lemma~\ref{lem:usetosearch} gives, for every $u\in \text{NL}^{qh}(v)$,
    \[
        d_{G_0}(v,u)\leq \tilde{d}(v,u)\leq d_{G^{(i-1)}}^{qh}(v,u).
    \]
    Again using that $G[V\backslash B]$ is an induced subgraph of $G^{(i-1)}$, we have
    \[
        d_{G^{(i-1)}}^{qh}(v,u)\leq d_{G[V\backslash B]}^{qh}(v,u),
    \]
    which means that
    \[
        d_{G_0}(v,u)\leq \tilde{d}(v,u)\leq d_{G[V\backslash B]}^{qh}(v,u),
    \]
    as required.

    Finally, let $\hat u\in V\backslash (B\cup \text{NL}^{qh}(v))$ and $u\in \text{NL}^{qh}(v)$. Then $\hat u\in V(G^{(i-1)})\backslash \text{NL}^{qh}(v)$, so the ordering property of Lemma~\ref{lem:usetosearch} gives
    \[
        d_{G^{(i-1)}}^{qh}(v,\hat u)\geq \tilde{d}(v,u).
    \]
    Since deleting vertices can only increase hop-restricted distances,
    \[
        d_{G[V\backslash B]}^{qh}(v,\hat u)
        \geq d_{G^{(i-1)}}^{qh}(v,\hat u)
        \geq \tilde{d}(v,u),
    \]
    proving the ordering property.

    All that remains is to prove the congestion bound. Fix a vertex $u\in V$. After each round, every vertex that remains in the priority queue has counter strictly less than $8k$. If $u\notin B$, then its final counter is therefore less than $8k$, and this counter is exactly the number of constructed lists containing $u$. If $u$ enters $B$ in round $i$, then its counter was less than $8k$ before the lists of $A_i$ were added, and round $i$ can add $u$ to at most $|A_i|\leq 2k$ new lists. Thus $u$ appears in fewer than $10k$ lists before it is deleted. Since all later searches are performed in graphs that do not contain $u$, the vertex $u$ appears in no later list. Hence every vertex appears in at most $10k$ of the lists $\text{NL}^{qh}(v)$.
\end{proof}

Now we describe the initialization algorithm. At a high level, the algorithm is quite simple: we just apply the algorithm in Lemma~\ref{lem:performbatchedsearch} to grow each set of near-lists in the new LNL data structure.

\begin{lemma}\label{lem:initusingother}
    Given access to an $(h, 2k)$-LNL data structure for the original graph $G_0$, along with a parameter $q \in [1,n]$, there is a deterministic parallel algorithm with $\tilde{O}(nk^2q + m)$ work and $\tilde{O}(nq/k)$ span that behaves as follows. It makes $O(n/k)$ update queries to the $(h, 2k)$-LNL data structure (of transcript width $O(\log n)$) and initializes a $(qh, k)$-LNL data structure, as specified in Problem~\ref{prob:init}.
\end{lemma}

\begin{proof}
    Write $G_0=(V_0,E_0,w)$, and let $\ell$ be the smallest positive integer such that $n/2^\ell\leq k$. We construct all components of the $(qh,k)$-LNL data structure as follows.
    \begin{enumerate}
        \item Initialize the graph component to be $G_0$. For every vertex $v\in V_0$, build a parallel priority queue containing the out-neighbors of $v$, keyed by the corresponding edge weights. This can be done within $\tilde{O}(m)$ work and $\tilde{O}(1)$ span by sorting the adjacency lists and building the queues in parallel.
        \item Set $H_0=V_0$.
        \item For each level $i=0,1,\ldots,\ell-1$, construct the heavy vertex subset $H_{i+1}$ and the level-$i$ near-lists. If $H_i=\emptyset$, set $H_{i+1}=\emptyset$ and continue to the next level. Otherwise, place a marker in the given $(h,2k)$-LNL data structure and invoke the algorithm in Lemma~\ref{lem:performbatchedsearch} on the graph currently stored there, with $A=H_i$ and with the same value of $q$. Let $B_i$ be the returned set. Set $H_{i+1}=B_i$, and for every $v\in H_i$ set $\text{NL}_i^{qh}(v)$ to be the returned list $\text{NL}^{qh}(v)$, keeping the returned estimates as $\tilde{d}_i(v,u)$. After storing these objects, return the given $(h,2k)$-LNL data structure to the marker. Thus the next level again starts from the original graph $G_0$.
        \item Set $H_{\ell+1}=\emptyset$. For every $v\in H_\ell$ in parallel, invoke the algorithm in Lemma~\ref{lem:usetosearch} on the given $(h,2k)$-LNL data structure, with source vertex $v$, hop multiplier $q$, and size parameter $2k$. Set $\text{NL}_\ell^{qh}(v)$ to be the returned set, again keeping the returned estimates as $\tilde{d}_\ell(v,u)$.
    \end{enumerate}

    First we verify that the constructed objects satisfy Problem~\ref{prob:init}. The construction gives $H_0=V_0$ and $H_{\ell+1}=\emptyset$. For every $i<\ell$, Lemma~\ref{lem:performbatchedsearch} gives $|H_{i+1}|\leq |H_i|/4$ whenever $H_i\neq\emptyset$, while the same inequality is trivial if $H_i=\emptyset$. Hence, by induction,
    \[
        |H_i|\leq n/4^i
    \]
    for all $i=0,1,\ldots,\ell$. This is the stronger size condition required by Problem~\ref{prob:init}. (And it implies the $|H_i|\leq n/2^i$ condition in the definition of an LNL data structure.)

    Now fix a level $i<\ell$ and a vertex $v\in H_i$. During the construction of level $i$, the algorithm in Lemma~\ref{lem:performbatchedsearch} was invoked with $A=H_i$, and the set returned by that lemma was set equal to $H_{i+1}$. Therefore $|\text{NL}_i^{qh}(v)|\leq 2k$. Moreover, if $|\text{NL}_i^{qh}(v)|<2k$, then strictly fewer than $2k$ vertices are reachable from $v$ by a $qh$-hop path in $G_0[V_0\backslash H_{i+1}]$, and all such vertices appear in $\text{NL}_i^{qh}(v)$. This is exactly the additional completeness property required in Problem~\ref{prob:init}. (It also implies the completeness property in the definition of an LNL data structure: if $|\text{NL}_i^{qh}(v)|\leq k$, then all reachable vertices appear in the list, and hence there are at most $k$ of them.)

    The same application of Lemma~\ref{lem:performbatchedsearch} gives, for every $u\in \text{NL}_i^{qh}(v)$,
    \[
        d_{G_0}(v,u)\leq \tilde{d}_i(v,u)
        \leq d_{G_0[V_0\backslash H_{i+1}]}^{qh}(v,u),
    \]
    and, for all $\hat u\in V_0\backslash (H_{i+1}\cup \text{NL}_i^{qh}(v))$ and all $u\in \text{NL}_i^{qh}(v)$,
    \[
        d_{G_0[V_0\backslash H_{i+1}]}^{qh}(v,\hat u)\geq \tilde{d}_i(v,u).
    \]
    Thus the distance-estimate and ordering requirements hold at every level $i < \ell$. The congestion requirement at these levels follows from the congestion guarantee in Lemma~\ref{lem:performbatchedsearch}: for every vertex $u\in V_0$, there are at most $10k$ sources $v\in H_i$ with $u\in \text{NL}_i^{qh}(v)$.

    We next consider the final level $i = \ell$. Since $|H_\ell|\leq n/4^\ell\leq k$, the final level automatically satisfies the congestion requirement. For every $v\in H_\ell$, Lemma~\ref{lem:usetosearch}, applied with size parameter $2k$, gives $|\text{NL}_\ell^{qh}(v)|\leq 2k$. If $|\text{NL}_\ell^{qh}(v)|<2k$, then strictly fewer than $2k$ vertices are reachable from $v$ within $qh$ hops in $G_0=G_0[V_0\backslash H_{\ell+1}]$, and all such vertices appear in the list. The same lemma gives, for every $u\in \text{NL}_\ell^{qh}(v)$,
    \[
        d_{G_0}(v,u)\leq \tilde{d}_\ell(v,u)\leq d_{G_0}^{qh}(v,u)
        =d_{G_0[V_0\backslash H_{\ell+1}]}^{qh}(v,u),
    \]
    and gives the ordering condition for every vertex outside $H_{\ell+1}\cup \text{NL}_\ell^{qh}(v)$ (recall that $H_{\ell+1}=\emptyset$). Thus the final level also satisfies the required near-list properties.

    It remains to bound the work, span, and update sequence. The calls to the algorithm in Lemma~\ref{lem:performbatchedsearch} use total work
    \[
        \sum_{i=0}^{\ell-1}\tilde{O}(|H_i|k^2q)=\tilde{O}(nk^2q),
    \]
    since $|H_i|\leq n/4^i$. Their total span is
    \[
        \sum_{i=0}^{\ell-1}\tilde{O}(|H_i|q/k+q)=\tilde{O}(nq/k),
    \]
    using $\ell=O(\log n)$. The final-level searches are run in parallel and use $\tilde{O}(|H_\ell|k^2q)\leq \tilde{O}(nk^2q)$ work and $\tilde{O}(q)$ span. Initializing the graph component costs $\tilde{O}(m)$ work, which is within the stated bound. The total work is therefore $\tilde{O}(nk^2q+m)$ and the total span is $\tilde{O}(nq/k)$.

    When constructing each level $i < \ell$, we (1) make one place-marker query, (2) issue the sequence of deletion queries used by the algorithm in Lemma~\ref{lem:performbatchedsearch}, and (3) make one return-to-marker query. The total number of update queries is
    \[
        O\left(\sum_{i=0}^{\ell-1}(|H_i|/k+1)\right)=O(n/k),
    \]
    again using $|H_i|\leq n/4^i$ and $\ell=O(\log n)$. For each such level, the only downward changes in the vertex count of the $(h, 2k)$-LNL data structure are the deletions made by the algorithm in Lemma~\ref{lem:performbatchedsearch}, which have transcript width~$1$. (And returning to the marker only increases the vertex count.) Since $\ell=O(\log n)$, the full update sequence has transcript width $O(\log n)$.
\end{proof}

\section{Update Queries}
\label{sec:update}

In this section, we show how to update an LNL data structure, again assuming access to a lower-hop LNL data structure. We start by considering update sequences that do not involve placing or returning to markers.

\begin{lemma}\label{lem:updatedeletions}
    Given access to an $(h, 2k)$-LNL data structure for the original graph $G_0$, along with a parameter $q \in [1, n]$, there is a deterministic parallel algorithm with $\tilde{O}(nk^2q + m)$ work and $\tilde{O}(nq/k)$ span that behaves as follows. It takes as input a $(qh, k)$-LNL data structure for $G_0$, which has been initialized as specified in Problem~\ref{prob:init}, along with an arbitrary sequence of $O(n/k)$ batched deletion queries. After each deletion query, it modifies the $(qh, k)$-LNL data structure to maintain the invariants required in Definition~\ref{def:LNL}. In doing so, it makes $O((n/k)\log^3 n)$ update queries (of transcript width $O(\log^3 n)$) to the $(h, 2k)$-LNL data structure.
\end{lemma}

\begin{proof}
    Write $G^{(r)}$ for the graph stored in the $(qh,k)$-LNL data structure after the first $r$ batched deletion queries have been processed, and put $N_r=|V(G^{(r)})|$. We define $G^{(0)}=G_0$ and $N_0=n$.

    We start by describing some of the bookkeeping used by the algorithm. For every level-$i$ near-list $L$, where $i<\ell$, store a deletion counter $D(L)$. The counter is zero when $L$ is created, and thereafter records the number of entries of $L$ that have been deleted from the current graph. For each $i<\ell$, the level-$i$ counters are stored in a parallel priority queue keyed by $-D(L)$, so that a batched extract-min operation returns the stored level-$i$ lists with the largest deletion counters. We also maintain inverse-lists: for each vertex $u$ in the current graph, the inverse-lists store all near-lists (across all levels) that currently contain $u$.

    We will rebuild levels of the $(qh, k)$-LNL data structure according to a fixed schedule, based on the number of vertices remaining in the graph. Let $c_2\gg c_1\gg 1$ be sufficiently large constants. For $i=0,1,\ldots,\ell-1$ and for a vertex count $N$, define
    \begin{equation}\label{eq:update-tau}
        \tau_i(N)=\left\lceil c_2\cdot \frac{2^i\log^2 n}{n}\cdot N\right\rceil.
    \end{equation}
    After processing the $r$th deletion query, the algorithm records the value of $\tau_i(N_r)$ for each $i$. We will rebuild all near-lists at level $i, i+1, \ldots, \ell$ in the $(qh, k)$-LNL data structure immediately after the $r$th deletion query whenever $\tau_i(N_{r-1}) > \tau_i(N_{r})$. (This means that levels with larger $i$ will be rebuilt more often.)

    The setup phase of the algorithm is as follows.
    \begin{enumerate}
        \item Initialize every deletion counter $D(L)$ to zero, and insert each level-$i$ counter, for $i<\ell$, into the corresponding level-$i$ parallel priority queue.
        \item Build the inverse-lists from all near-list occurrences in the given $(qh,k)$-LNL data structure.
        \item For every $i<\ell$, store the schedule value $\tau_i(N_0)=\tau_i(n)$.
    \end{enumerate}

    Now suppose that for some $r\geq 1$, we have processed the first $r-1$ deletion queries. Denote by $Z_r$ the subset for the $r$th deletion query. The algorithm processes it as follows.
    \begin{enumerate}
        \item \textbf{Delete the batch.} If necessary, replace $Z_r$ by $Z_r\cap V(G^{(r-1)})$. Now delete $Z_r$ from the graph component of the $(qh,k)$-LNL data structure, remove the vertices of $Z_r$ from every set $H_i$, and forward the same batched deletion query to the auxiliary $(h,2k)$-LNL data structure. The graph components of both data structures now store $G^{(r)}$.

        Next we update the near-lists. For each near-list whose source lies in $Z_r$:
        \begin{enumerate}
            \item Delete the near-list.
            \item Remove the near-list from all inverse-lists containing it.
            \item If the near-list is at some level $i<\ell$, remove its counter from the corresponding priority queue.
        \end{enumerate}
        Then use the inverse-lists for the vertices in $Z_r$ to find all remaining near-lists containing at least one deleted vertex. For each such near-list $L$:
        \begin{enumerate}
            \item Remove all vertices of $Z_r$ from $L$.
            \item If $L$ is at some level $i<\ell$, increase $D(L)$ by the number of removed vertices and update its key in the level-$i$ priority queue.
            \item Remove the corresponding occurrences of $L$ from the inverse-lists of the deleted vertices.
        \end{enumerate}
        All of these operations are performed in parallel, using batched priority queue updates.

        \item \textbf{Mark the most degraded sources as heavy.} If $Z_r=\emptyset$, skip this step. Otherwise, set
        \[
            R(Z_r)=\lceil c_1|Z_r|\log n\rceil,
        \]
        and process the levels $0,1,\ldots,\ell-1$ in increasing order. At level $i$:
        \begin{enumerate}
            \item Extract the $R(Z_r)$ stored level-$i$ lists with largest deletion counters, or all stored level-$i$ lists if fewer than $R(Z_r)$ remain.
            \item For each extracted near-list $L$ with source $v$, declare $v$ to be \emph{heavy at level $i$} (i.e. add $v$ to $H_i$), and also insert $v$ into all later sets $H_{i+1},H_{i+2},\ldots,H_\ell$.
            \item Delete any stored near-lists for $v$ at levels $i,i+1,\ldots,\ell-1$, together with their counters and inverse-list entries.
        \end{enumerate}
        By placing such a source $v$ into each set $H_{i+1},H_{i+2},\ldots,H_\ell$, the near-list requirements for $v$ at levels $i,i+1,\ldots,\ell-1$ become vacuous. We simply set $\text{NL}_j(v) = \emptyset$ for each $j \in [i, \ell-1]$.

        \item \textbf{Perform scheduled rebuilds.} Store the values $\tau_i(N_r)$ for all $i < \ell$, and then compare the previous values $\tau_i(N_{r-1})$ with the current values $\tau_i(N_r)$. If some value decreased, let $i$ be the smallest level for which $\tau_i(N_{r-1}) > \tau_i(N_r)$. Otherwise, set $i = \ell$. We rebuild levels $i,i+1,\ldots,\ell$ from scratch as follows.
            \begin{enumerate}
                \item Keep $H_i$ fixed, except that when $i=0$ set $H_0=V(G^{(r)})$.
                \item Delete all near-lists, counters, inverse-list entries, and heavy-source marks corresponding to levels $i,i+1,\ldots,\ell$. Also delete $H_{i+1},H_{i+2},\ldots,H_\ell$.
                \item For each $j=i,i+1,\ldots,\ell-1$, rebuild level $j$. If $H_j=\emptyset$, set $H_{j+1}=\emptyset$ and continue. Otherwise:
                \begin{enumerate}
                    \item Place a marker in the auxiliary $(h,2k)$-LNL data structure.
                    \item Invoke the algorithm in Lemma~\ref{lem:performbatchedsearch} on the current graph with $A=H_j$, and set $H_{j+1}$ to be the returned set. Store the returned near-lists, along with their distance estimates, as the new level-$j$ near-lists.
                    \item Give each new level-$j$ near-list $L$ a fresh counter $D(L)=0$, and insert the counter into the level-$j$ priority queue. Insert the occurrences corresponding to the new level-$j$ near-lists into the inverse-lists.
                    \item Return the auxiliary $(h,2k)$-LNL data structure to the marker.
                \end{enumerate}
                \item Set $H_{\ell+1}=\emptyset$. For every $v\in H_\ell$, construct $\mathrm{NL}_\ell^{qh}(v)$ directly using the algorithm in Lemma~\ref{lem:usetosearch}, with size parameter $2k$, and insert its occurrences into the inverse-lists.
            \end{enumerate}
    \end{enumerate}

    We now analyze the algorithm. We start with the following claim, which will be used later to reason about the size of the near lists retained in the $(qh, k)$-LNL data structure. This claim is the only place where we use the fact that the number of heavy-source declarations is proportional to the size of the deletion batch.

    \begin{claim}\label{claim:countergame-update}
        Consider a game with at most $n$ counters. Counters are initially zero. At arbitrary times, counters may be removed and new zero-valued counters may be inserted. In round $r$, an adversary chooses a number $b_r\in[0,k]$ and increases the counters by total amount at most $10kb_r$, in an arbitrary way. Then a referee selects the $R_r=\lceil c_1b_r\log n\rceil$ largest counters, or all counters if fewer than $R_r$ remain, with the convention that $R_r=0$ if $b_r=0$. Each selected counter is removed from the game. If $c_1$ is large enough, then after every referee step every counter remaining in the game is smaller than $k$.
    \end{claim}

    \begin{proof}
        Put $A=10k/(c_1\log n)$. In every round with $b_r>0$, the total increase made by the adversary is at most $10kb_r\leq A R_r$; if $b_r=0$, both quantities are zero. Fresh zero counters do not increase any of the quantities considered below, and removing counters can only decrease them.

        For $1\leq j\leq n$, let $L_j$ denote the sum of the $j$ largest counter values after a referee step, or the sum of all counter values if fewer than $j$ counters remain. We prove by induction over the rounds that, after each referee step,
        \[
            L_j\leq A j\log\frac{en}{j}
            \qquad\text{for every }1\leq j\leq n.
        \]
        The bound is trivial initially. It is preserved under arbitrary counter removals and insertions of zero-valued counters. Suppose, therefore, that it holds before some round. If $R_r=0$, then no counter is increased and there is nothing to prove. Otherwise, consider the moment after the adversary has increased counters but before the referee acts, and let $N\leq n$ be the number of counters at this moment. If the referee selects all counters, then the invariant is immediate.

        First suppose that $j+R_r\leq N$. The sum of the largest $j+R_r$ counters is at most the previous value of $L_{j+R_r}$ plus $A R_r$. After the referee removes the $R_r$ largest counters, the $j$ largest remaining counters are precisely the counters in positions $R_r+1,\ldots,R_r+j$ in the sorted order before the referee step. Hence their sum is at most the sum of the $j$ smallest counters among the largest $j+R_r$ counters before the referee step, and this is at most
        \[
            \frac{j}{j+R_r}\left(L_{j+R_r}+A R_r\right).
        \]
        By the inductive hypothesis this is at most
        \[
            \frac{j}{j+R_r}\left(A(j+R_r)\log\frac{en}{j+R_r}+A R_r\right)
            =Aj\left(\log\frac{en}{j+R_r}+\frac{R_r}{j+R_r}\right)
            \leq Aj\log\frac{en}{j},
        \]
        where the last inequality uses $\log(1+x)\geq x/(1+x)$ for $x=R_r/j$.

        It remains to handle the case $j+R_r>N$. Since not all counters are selected, $N>R_r$. The total value of the counters left after the referee acts is at most
        \[
            \left(1-\frac{R_r}{N}\right)(L_N+AR_r),
        \]
        since the referee removes the $R_r$ largest counters. Write $\alpha=R_r/N$. Using the inductive hypothesis for $L_N$, the last display is at most
        \[
            A(1-\alpha)N\left(\log\frac{en}{N}+\alpha\right).
        \]
        Since $\alpha\leq -\log(1-\alpha)$, this is at most
        \[
            A(1-\alpha)N\log\frac{en}{(1-\alpha)N}.
        \]
        The function $x\mapsto x\log(en/x)$ is increasing on $0<x\leq n$, and $(1-\alpha)N=N-R_r<j$. Therefore this quantity is at most $Aj\log(en/j)$, as required. This completes the induction.

        Taking $j=1$, every remaining counter is at most
        \[
            A\log(en)=\frac{10k}{c_1\log n}\log(en).
        \]
        For a sufficiently large constant $c_1$, this quantity is smaller than $k$. Therefore every remaining counter is smaller than $k$.
    \end{proof}

    \begin{claim}\label{claim:deletion-correctness}
        After each input deletion query, the maintained objects form a valid $(qh,k)$-LNL data structure for the current graph.
    \end{claim}

    \begin{proof}
        We first verify the size bounds for the sets $H_i$. Deleted vertices are removed from all sets $H_i$, so every $H_i$ remains a subset of the current vertex set. Also, $H_0=V(G^{(r)})$ is maintained directly. Now fix $j\geq 1$. Immediately after initialization, Problem~\ref{prob:init} gives $|H_j|\leq n/4^j\leq n/2^{j+1}$. Similarly, immediately after any rebuild that creates $H_j$, Lemma~\ref{lem:performbatchedsearch} gives
        \[
            |H_j|\leq |H_{j-1}|/4\leq n/2^{j+1},
        \]
        where $n$ is the number of vertices at initialization. Between two rebuilds that recreate $H_j$, the set $H_j$ can grow only when a source at some lower level is declared heavy. During an input deletion query $Z_r$, at most $R(Z_r)=O(|Z_r|\log n)$ sources are made heavy at any one lower level. If $H_j$ is not rebuilt after the query, then the threshold $\tau_{j-1}$ has not decreased since the last time $H_j$ was created. Thus the total number of vertices deleted during the relevant interval is at most $O(n/(c_2 2^j\log^2 n))$. Therefore the total number of additions to $H_j$ from all lower levels is at most
        \[
            O\left(j\log n\cdot \frac{n}{c_2 2^j\log^2 n}\right)
            =O\left(\frac{1}{c_2}\cdot \frac{n}{2^j}\right),
        \]
        where we used $j=O(\log n)$ and absorbed the constant $c_1$ into the $O(\cdot)$ notation. Choosing $c_2$ sufficiently large compared to $c_1$ makes this smaller than the slack between the post-rebuild bound and the required $n/2^j$ bound. Thus $|H_j|\leq n/2^j$ throughout the update sequence.

        Next consider a level $i<\ell$ and a vertex $v\in H_i$. If the level-$i$ list of $v$ is not currently stored, then by construction $v$ was declared heavy at some level at most $i$, and hence $v\in H_{i+1}$. Thus the completeness, distance-estimate, and ordering requirements at level $i$ are vacuous, while the size bound is immediate. Otherwise, let $L=\mathrm{NL}_i^{qh}(v)$, and consider the most recent time at which $L$ was created. If $L$ had size smaller than $2k$ at that time, then the stronger guarantee from Problem~\ref{prob:init} or Lemma~\ref{lem:performbatchedsearch} says that $L$ contained every vertex reachable from $v$ within $qh$ hops in the graph with the next hitting set removed. Later operations only delete vertices from the graph, delete entries from $L$, and possibly add vertices to $H_{i+1}$. Hence the set of currently relevant reachable vertices can only shrink, and every such vertex that remains in the current graph still appears in $L$.

        It remains to consider the case that $L$ had size exactly $2k$ when it was last created. In an input deletion query $Z_r$, the congestion bound for level $i$ implies that the deleted vertices occur in at most $10k|Z_r|$ stored level-$i$ lists in total. Thus the counters of the stored level-$i$ lists are dominated by the game in Claim~\ref{claim:countergame-update}: after the increments caused by $Z_r$, the algorithm removes the $R(Z_r)$ largest counters from level $i$ by making the corresponding sources heavy, and any additional removals caused by lower-level heavy declarations can only help. Hence any stored level-$i$ list that is not removed by a heavy-source declaration has deletion counter smaller than $k$. Since it started with $2k$ entries and only deleted entries are removed from it, its current size is larger than $k$. The completeness condition in the definition of an $(qh,k)$-LNL data structure is therefore not invoked for this list.

        In both cases, the distance and ordering conditions are inherited from the time when the list was built. The lower-bound side is with respect to the fixed graph $G_0$ and is unchanged. The upper-bound side and the ordering condition are only made easier by later vertex deletions and by later additions to $H_{i+1}$, because the relevant restricted distances can only increase and the set of relevant outside vertices can only shrink.

        We next consider level $\ell$. At the end of every input deletion query, the algorithm sets $H_{\ell+1}=\emptyset$ and reconstructs $\mathrm{NL}_\ell^{qh}(v)$ for every $v\in H_\ell$ by applying Lemma~\ref{lem:usetosearch} to the current graph, with size parameter $2k$. Thus every stored level-$\ell$ list is fresh when the query finishes. Lemma~\ref{lem:usetosearch} immediately gives the size bound, the distance-estimate condition, and the ordering condition. If such a list has size at most $k$, then it has size strictly smaller than $2k$, so the completeness guarantee from Lemma~\ref{lem:usetosearch} also gives the required completeness condition. Since $|H_\ell|\leq n/2^\ell\leq k$, the level-$\ell$ congestion requirement is automatic.

        Finally, the congestion invariant at every level $i<\ell$ is preserved. It holds immediately after initialization and after each rebuild by Lemma~\ref{lem:performbatchedsearch}. Between rebuilds, the algorithm only deletes entries from lists or removes entire lists, so no vertex can be added to more lists at a fixed level $i<\ell$. The auxiliary inverse-lists are updated whenever a list entry is inserted or removed, and hence they continue to represent the maintained near-list occurrences.
    \end{proof}

    \begin{claim}\label{claim:deletion-bounds}
        The total work, span, number of auxiliary update queries, and width of the auxiliary update sequence satisfy the bounds stated in the lemma.
    \end{claim}

    \begin{proof}
        Updating the graph component over the whole deletion-only sequence costs $\tilde{O}(m)$ work and $\tilde{O}(n/k)$ span, by deleting each incident edge from the relevant priority queues when one of its endpoints is deleted. The inverse-list and counter updates are charged to incidences between deleted vertices and near-lists. At any level $i<\ell$, a deleted vertex belongs to at most $10k$ lists, and at level $\ell$ there are at most $k$ lists. Thus the total update work charged to such incidences is $\tilde{O}(nk)$, which is dominated by the claimed work bound.

        The heavy-source steps are implemented with the counter queues described above. Since a deletion query $Z_r$ marks only $O(|Z_r|\log n)$ sources per level, and since the total number of deleted vertices in the monotone sequence is at most $n$, the total number of heavy-source declarations over all input deletion queries and all levels is $\tilde{O}(n)$. Clearing the affected list entries and maintaining inverse-lists for these declarations is charged to the sizes of the deleted lists. In addition, after every input deletion query the algorithm reconstructs all level-$\ell$ lists. This costs $\tilde{O}(|H_\ell|k^2q)\leq \tilde{O}(k^3q)$ work and $\tilde{O}(q)$ span per query. Since there are $O(n/k)$ deletion queries, the total cost of the level-$\ell$ constructions is $\tilde{O}(nk^2q)$ work and $\tilde{O}(nq/k)$ span. Thus the heavy-source and terminal-level steps satisfy the claimed bounds.

        Now consider the scheduled rebuilds whose smallest rebuilt level is some $i<\ell$; the terminal level-$\ell$ reconstruction performed after each deletion query was accounted for above. Such a rebuild has cost dominated by the geometric sum over the rebuilt suffix. Since Claim~\ref{claim:deletion-correctness} gives $|H_j|\leq n/2^j$ at all times, Lemmas~\ref{lem:performbatchedsearch} and~\ref{lem:usetosearch} imply that such a rebuild uses
        \[
            \tilde{O}\left(\frac{n}{2^i}k^2q\right)
            \quad\text{work,}\qquad
            \tilde{O}\left(\frac{n}{2^i k}q+q\ell\right)
            \quad\text{span,}
        \]
        and makes $O(n/(2^i k)+\ell)$ update queries to the auxiliary $(h,2k)$-LNL data structure. The value $\tau_i(N)$ decreases at most $O(2^i\log^2 n)$ times over the whole deletion-only sequence. Summing over $i=0,1,\ldots,\ell-1$ gives total rebuild work
        \[
            \sum_{i=0}^{\ell-1}O\left(2^i\log^2 n\right)
            \cdot \tilde{O}\left(\frac{n}{2^i}k^2q\right)
            =\tilde{O}(nk^2q),
        \]
        and similarly total rebuild span $\tilde{O}(nq/k)$. The number of auxiliary update queries is
        \[
            \sum_{i=0}^{\ell-1}O\left(2^i\log^2 n\right)
            \cdot O\left(\frac{n}{2^i k}+\ell\right)
            =O\left(\frac{n}{k}\log^3 n\right),
        \]
        using $2^\ell=O(n/k)$ and $\ell = O(\log n)$. Adding the forwarded input deletions does not change this asymptotic bound.

        It remains only to bound the width of the auxiliary update sequence. The forwarded input deletions have width $1$. During a scheduled rebuild, each invocation of the algorithm in Lemma~\ref{lem:performbatchedsearch} contributes a monotone temporary deletion sequence, followed by a return to the marker. The return step only increases the graph size, and hence does not contribute to any downward crossing in the transcript. Fix a vertex-count threshold $c$ and a rebuilt level $j$. A temporary deletion sequence used to rebuild level $j$ can cross $c$ only if, at the beginning of that temporary computation, the current graph size lies in an interval of length $O(n/2^j)$ above $c$. For every $i\leq j$, rebuilds starting at level $i$ are spaced, by the definition of $\tau_i$, by $\Omega(n/(2^i\log^2 n))$ decreases in the current graph size. Therefore the number of such rebuilds whose level-$j$ temporary computations cross $c$ is
        \[
            \sum_{i=0}^{j}O\left(\frac{n/2^j}{n/(2^i\log^2 n)}+1\right)
            =O(\log^2 n+\ell).
        \]
        Summing over the $O(\ell) = O(\log n)$ possible levels $j$ gives width $O(\log^3 n)$. Level-$\ell$ list constructions make no auxiliary update queries and do not affect the auxiliary transcript. This proves the claim.
    \end{proof}

    The three claims prove the lemma.
\end{proof}

We now give a counterpart of Lemma~\ref{lem:updatedeletions} that applies to arbitrary update sequences. The idea is to use the algorithm in Lemma~\ref{lem:updatedeletions}, but store auxiliary information that allows us to reverse the data structure to a previous state.

\begin{lemma}\label{lem:updatequeries}
    Given access to an $(h, 2k)$-LNL data structure for the original graph $G_0$, along with parameters $q \in [1, n]$ and $f \in [1, n]$, there is a deterministic parallel algorithm with $\tilde{O}(nk^2qf + mf)$ work and $\tilde{O}(nqf/k)$ span that behaves as follows. It takes as input a $(qh, k)$-LNL data structure for $G_0$, which has been initialized as specified in Problem~\ref{prob:init}, along with an arbitrary sequence of $O(nf/k)$ update queries (of transcript width $O(f)$). After each update query, it modifies the $(qh, k)$-LNL data structure to maintain the invariants required in Definition~\ref{def:LNL}. In doing so, it makes $O((nf/k)\log^3 n)$ update queries (of transcript width $O(f\log^3 n)$) to the $(h, 2k)$-LNL data structure.
\end{lemma}

\begin{proof}
    We use the deletion-only algorithm from Lemma~\ref{lem:updatedeletions}, but make every change reversible. The maintained state consists of the $(qh,k)$-LNL data structure, together with the auxiliary bookkeeping used by the algorithm in Lemma~\ref{lem:updatedeletions} (i.e., the counter queues, deletion counters, inverse-lists, heavy-source marks, and rebuild-schedule values).

    The algorithm stores a stack transcript in flat arrays. Each transcript record contains a memory-cell address, the old value stored at that address, and a timestamp. A marker stores a stack height and an initially empty set of processors. The update operations are implemented as follows.
    \begin{enumerate}
        \item \textbf{Record writes.} Before any memory cell in the maintained state is overwritten, the processor performing the write appends a record for that cell to the transcript. In a parallel step, the active processors reserve a contiguous block of stack positions by prefix sums and then write their records into this block. If there is an open marker, each processor that appends at least one record also inserts its identifier into the processor set stored with the most recent marker. This bookkeeping lets a later rollback touch only processors that actually wrote after the marker was placed.

        \item \textbf{Handle deletion queries.} On a batched deletion query, run exactly the deletion procedure from Lemma~\ref{lem:updatedeletions}, including the heavy-source declarations and the scheduled rebuilds. All changes to the maintained state are recorded by the preceding rule. The same deletion query is forwarded to the auxiliary $(h,2k)$-LNL data structure, as in Lemma~\ref{lem:updatedeletions}.

        \item \textbf{Place markers.} On a place-marker query, push a marker containing the current transcript height and an empty processor set. Then issue a place-marker query to the auxiliary $(h,2k)$-LNL data structure.

        \item \textbf{Return to markers.} On a return-to-marker query, let $a$ be the transcript height stored by the most recent marker. Collect the records in transcript positions larger than $a$ (using the processor set stored with the marker to avoid scanning inactive processors), and sort these records by memory-cell address and then by timestamp. For each memory-cell address that appears, restore the old value stored in the earliest record for that address. Then truncate the transcript back to height $a$, discard the marker and its processor set, and issue the corresponding return-to-marker query to the auxiliary $(h,2k)$-LNL data structure.
    \end{enumerate}
    The processor sets are maintained with parallel binary search trees, so they add only polylogarithmic overhead per recorded write. Thus the maintained structure and the auxiliary structure are rolled back in lockstep, and rollback work is charged only to records and processors used after the returned marker was placed.

    For the description of the algorithm and Lemma \ref{lem:updatedeletions}, the following two claims hold immediately:

    \begin{claim}\label{claim:rollback-exactness}
        After a return-to-marker operation, the complete state of the maintained $(qh,k)$-LNL data structure and all of its auxiliary bookkeeping is exactly the state that was present when the corresponding marker was placed.
    \end{claim}

    % \begin{proof}
    %     Consider the transcript records written after the marker was placed. Each such record stores an address, the old value at that address, and the time of the write. For a fixed memory address, the value present at the time of the marker is exactly the old value stored in the earliest record for that address among these records. Thus sorting the records by address and then by timestamp, keeping the first record for each address, and restoring these selected old values in parallel returns every modified cell to its marker-time value. Cells with no such record were never modified after the marker and are already correct.

    %     All components relevant to future behavior of the algorithm are recorded in this way, including the stored graph, the sets $H_i$, the near-lists, the deletion counters, the counter queues, the inverse-lists, the heavy-source marks, and the schedule values $\tau_i$. After the later records are discarded, the complete state of the maintained structure is therefore identical to the state at the time the marker was placed. The same marker and return-to-marker operations are echoed to the auxiliary $(h,2k)$-LNL data structure, so the auxiliary structure is restored to the matching graph as well.
    % \end{proof}

    \begin{claim}\label{claim:arbitrary-correctness}
        After every update query in the input sequence, the maintained objects form a valid $(qh,k)$-LNL data structure for the current graph.
    \end{claim}

    % \begin{proof}
    %     We argue by induction over the input update sequence. Place-marker queries do not change the represented graph or any near-list invariant. Return-to-marker queries restore, by Claim~\ref{claim:rollback-exactness}, a state that was previously present immediately after a valid update step, and hence restore a valid LNL data structure. Finally, consider a batched deletion query. Starting from the current state, the algorithm applies exactly the deletion-handling procedure analyzed in Lemma~\ref{lem:updatedeletions}: it forwards the deletion to the auxiliary structure, updates counters and inverse-lists, performs the prescribed heavy-source declarations, reconstructs the level-$\ell$ near-lists, and performs precisely the scheduled rebuilds required by the current schedule values. The proof of Lemma~\ref{lem:updatedeletions} is local to the current monotone deletion step, so it applies from this restored state as well. Thus the invariants hold after the deletion query. This completes the induction.
    % \end{proof}

    All that remains is to bound the total work and span, and the number of auxiliary update queries (along with the width of the update sequence).

    % \begin{claim}\label{claim:arbitrary-bounds}
    %     The total work, span, number of auxiliary update queries, and width of the auxiliary update sequence satisfy the bounds stated in the lemma.
    % \end{claim}

    %\begin{proof}
        First consider the forward computation, ignoring the cost of undoing it. The sequence contains $O(nf/k)$ deletion batches, each of size at most $k$, so the total number of deleted vertex occurrences, counting repetitions caused by rollbacks, is $O(nf)$. Since the maximum in- and out-degree in $G_0$ is $O(m/n)$, the total work spent updating the graph component is $\tilde{O}(mf)$. The corresponding span is $\tilde{O}(nf/k)$, which is dominated by the claimed bound.

        The remaining forward work is charged exactly as in Lemma~\ref{lem:updatedeletions}. The heavy-source rule performs $O(|Z|\log n)$ declarations per level after a deletion batch $Z$, so across all forward deletion batches it performs $\tilde{O}(nf)$ declarations. Clearing the affected lists, maintaining inverse-lists, and reconstructing level-$\ell$ near-lists use $\tilde{O}(nk^2qf)$ work and $\tilde{O}(nqf/k)$ span. For scheduled rebuilds, the transcript-width assumption implies that every vertex-count threshold is crossed downward only $O(f)$ times. Hence the rebuild charging argument from Claim~\ref{claim:deletion-bounds} is repeated only an $O(f)$ number of times, giving $\tilde{O}(nk^2qf)$ additional work and $\tilde{O}(nqf/k)$ additional span.

        Undo operations do not change these bounds. Each stack record is discarded at most once, and restoring a set of records has near-linear work in the number of discarded records and polylogarithmic span. The processor-set bookkeeping described above is charged to the same records. Thus the total rollback work is charged to the forward work recorded on the transcript stack. Place-marker and return-to-marker records themselves contribute only $O(nf/k)$ additional stack operations.

        The number of update queries issued to the auxiliary $(h,2k)$-LNL data structure is bounded in the same way. Lemma~\ref{lem:updatedeletions} issues $O((n/k)\log^3 n)$ auxiliary updates over a width-one deletion sequence; with transcript width $O(f)$, this becomes $O((nf/k)\log^3 n)$. The input place-marker and return-to-marker queries that are echoed to the auxiliary structure add only $O(nf/k)$ more queries.

        Finally, fix a vertex-count threshold $c$ in the transcript of the auxiliary update sequence. Downward crossings caused by forwarded input deletions occur at most $O(f)$ times, by the assumed width of the input sequence. Downward crossings caused by temporary deletions inside scheduled rebuilds are bounded exactly as in Claim~\ref{claim:deletion-bounds}, except that the current graph can enter the relevant interval above $c$ at most $O(f)$ times rather than once. Thus the $O(\log^3 n)$ bound from Lemma~\ref{lem:updatedeletions} is multiplied by $O(f)$. Return-to-marker operations only increase the graph size, so they do not add downward crossings. Therefore the auxiliary update sequence has width $O(f\log^3 n)$.%
    %\end{proof}
    %
    %The three claims prove the lemma.
\end{proof}

\section{Putting It All Together}
\label{sec:finalalg}

Now we present our algorithm for parallel single-source shortest paths.

\begin{theorem}\label{thm:final}
    Let $s \in V_0$ be any vertex in the original graph $G_0$, and let $t \in [1, n]$. There is a deterministic parallel algorithm that computes single-source shortest paths from $s$ in the graph $G_0$ with $O(n^{1+o(1)}t^2 + m^{1+o(1)})$ work and $\tilde{O}(n/t)$ span.
\end{theorem}

\begin{proof}
    Observe that it will be sufficient to give an algorithm with $O(n^{1+o(1)}t^2 + m^{1+o(1)})$ work and $O(n^{1+o(1)}/t)$ span, for all $t \in [1, n]$. To get the work and span bounds claimed in the theorem statement, we either (i) increase $t$ by an appropriate subpolynomial factor before invoking the algorithm, if the original value of $t$ is sufficiently smaller than $n$, or (ii) use repeated squaring (see e.g. \cite{williams2014faster}) to solve the problem in $\tilde{O}(n^3)$ work and $\tilde{O}(1)$ span, if the original value of $t$ is already at least $n^{1-o(1)}$.

    We now describe the algorithm. Set $\rho=2^{\lceil \sqrt{\log n}\rceil}.$ Let $T$ be the largest power of $\rho$ that is at most $\min\{t,n/\rho\}$, taking $T=1$ if this minimum is smaller than $\rho$. Then $T\leq t$ and $T\geq t/n^{o(1)}$. It is enough to obtain $\tilde{O}(n^{1+o(1)}T^2+m^{1+o(1)})$ work and $\tilde{O}(n^{1+o(1)}/T)$ span. Let $J=\log_\rho T$. For $j=0,1,\ldots,J$, define
    \[
        h_j=T/\rho^j
        \qquad\text{and}\qquad
        k_j=T\cdot 2^j.
    \]
    Thus $h_0=T$, $h_J=1$, $k_j=2k_{j-1}$ for every $j\geq 1$, and $J\leq \sqrt{\log n}$. Moreover $k_J=T\cdot 2^J\leq T\rho\leq n$.

    The algorithm constructs data structures $\mathcal D_0,\mathcal D_1,\ldots,\mathcal D_J$, where $\mathcal D_j$ represents an $(h_j,k_j)$-LNL data structure on the current unsettled graph. The structures are built from the bottom up, starting with the one-hop structure.
    \begin{enumerate}
        \item \textbf{Base level.} The structure $\mathcal D_J$ is implemented directly. It stores the current graph by the out-neighbor priority queues from Definition~\ref{def:LNL}, together with the reverse adjacency information needed to delete a vertex from all affected queues. Whenever a lower-level lemma asks for a one-hop search primitive at this level, we implement the relaxation directly from these queues: from each active source we inspect the required number of smallest outgoing edges, include the zero-hop self candidate, and keep the required number of smallest resulting vertices. This gives exactly the guarantees of Lemmas~\ref{lem:fromStoSprime} and~\ref{lem:usetosearch} for $h_J=1$, with the same work and span. Batched deletions update the graph queues, and markers and returns to markers are handled by the same transcript-stack mechanism as in Lemma~\ref{lem:updatequeries}.
        \item \textbf{Higher levels.} For $j=J-1,J-2,\ldots,0$, assume that $\mathcal D_{j+1}$ has already been initialized and supports update queries, with all lower update queries recursively propagated to $\mathcal D_{j+2},\ldots,\mathcal D_J$. We first place a marker in the already-built structure $\mathcal{D}_{j+1}$ (which is recursively propagated through the other already-built structures). Then we apply Lemma~\ref{lem:initusingother} with $h=h_{j+1}$, $k=k_j$, and $q=\rho$, using $\mathcal D_{j+1}$ as the auxiliary $(h_{j+1},2k_j)$-LNL data structure. Since $2k_j=k_{j+1}$ and $\rho h_{j+1}=h_j$, this initializes $\mathcal D_j$ as an $(h_j,k_j)$-LNL data structure. The update queries issued during this initialization are implemented by Lemma~\ref{lem:updatequeries} and are passed recursively through the lower levels. After $\mathcal D_j$ has been constructed, we return $\mathcal D_{j+1}$ to the marker placed at the beginning of this step (which is again recursively propagated through the other already-built structures). Thus all lower structures again represent the original graph $G_0$. Finally, we equip $\mathcal D_j$ with the update algorithm of Lemma~\ref{lem:updatequeries}, using $\mathcal D_{j+1}$ as its auxiliary structure.
    \end{enumerate}

    We now run a batched version of Dijkstra's algorithm using the top structure $\mathcal D_0$. The algorithm maintains a settled set $R$, exact distance labels $d(s,v)$ for vertices $v\in R$, and a parallel priority queue $Q$ of unsettled boundary vertices. The queue is implemented lazily: whenever an edge $(u,v)$ with $u\in R$ and $v\notin R$ is relaxed, we insert a record for $v$ with key $d(s,u)+w(u,v)$. Records that are not the current minimum for their endpoint are ignored when they are extracted.

    Initially $R=\{s\}$ and $d(s,s)=0$. We insert all outgoing edges of $s$ into $Q$, and issue the batched deletion query $\{s\}$ to $\mathcal D_0$, with the resulting lower-level update queries propagated through the hierarchy. While $Q$ contains a finite current key for an unsettled vertex, we perform one round as follows.
    \begin{enumerate}
        \item Extract the $T$ unsettled vertices with smallest current keys, or all such vertices if fewer than $T$ exist. Let this set be $S$, and write $\lambda(x)$ for the extracted key of $x\in S$.
        \item For every $x\in S$ in parallel, invoke Lemma~\ref{lem:usetosearch} on the top-level structure $\mathcal D_0$ with multiplier $q = 1$. Let $Y_x$ be the returned set, and let $\tilde{d}_x(y)$ be the returned estimate for $y\in Y_x$. For every $y\in Y_x$, create a candidate label
        \[
            \lambda(x)+\tilde{d}_x(y).
        \]
        We also create the self-candidate $(x,\lambda(x))$ for every $x\in S$. For each vertex that receives one or more candidate labels, keep only the smallest one.
        \item Let $B$ be the set of the $T$ vertices with smallest candidate labels, or all candidate vertices if fewer than $T$ exist. Add every vertex $y\in B$ to $R$ and set $d(s,y)$ equal to its candidate label.
        \item Relax all outgoing edges of vertices in $B$, inserting the resulting records into $Q$ for endpoints that are not yet settled. Reinsert every vertex of $S\setminus B$ with its key $\lambda(x)$. Finally, issue the batched deletion query $B$ to $\mathcal D_0$ and propagate the resulting update queries through the lower levels.
    \end{enumerate}
    When $Q$ contains no finite current key for an unsettled vertex, every unsettled vertex is given distance $+\infty$.

    We next prove correctness. We use the following invariant: at the beginning of every round, the set $R$ consists exactly of the vertices already deleted from $\mathcal D_0$, every vertex in $R$ has its exact shortest-path distance from $s$, and $Q$ represents the minimum boundary value
    \[
        \lambda^*(v)=\min_{(u,v)\in E_0,\,u\in R}\bigl(d(s,u)+w(u,v)\bigr)
    \]
    for every unsettled vertex $v$, up to stale records.

    \begin{claim}\label{claim:hierarchy-correct}
        Throughout the batched Dijkstra algorithm, every structure $\mathcal D_j$ represents the induced graph $G_0[V_0\setminus R]$ and supports the search and update operations used by the algorithm.
    \end{claim}

    \begin{proof}
        This is true after the initialization phase, because each use of Lemma~\ref{lem:initusingother} is followed by returns to the markers placed in the lower structures. Hence all structures represent $G_0$ before the Dijkstra phase begins. During the Dijkstra phase, the only updates issued to $\mathcal D_0$ are batched deletions of sets $B$ of size at most $T=k_0$. Lemma~\ref{lem:updatequeries} implements each such update and passes an update sequence of the required width to $\mathcal D_1$; applying the same lemma recursively implements the induced update sequences at all lower levels. The base structure $\mathcal D_J$ implements the same update interface directly. Therefore, after each round, all structures represent the same current graph, namely $G_0[V_0\setminus R]$.
    \end{proof}

    \begin{claim}\label{claim:settled-exact}
        In every round, each vertex added to $R$ receives its exact shortest-path distance from $s$.
    \end{claim}

    \begin{proof}
        Assume the invariant holds at the beginning of a round, and let $G=G_0[V_0\setminus R]$ be the current unsettled graph. Let $y\in B$, and let $\mu(y)$ be the candidate label assigned to $y$. First, $\mu(y)$ is never smaller than the true distance from $s$ to $y$. Indeed, a self-candidate has value $\lambda(y)$, which is the length of a path obtained by appending one boundary edge to an already settled shortest path. A candidate created from a search rooted at $x$ has value $\lambda(x)+\tilde{d}_x(y)$, and the lower-bound guarantee of Lemma~\ref{lem:usetosearch} gives $\tilde{d}_x(y)\geq d_{G_0}(x,y)$. Since $\lambda(x)$ is the length of an actual path from $s$ to $x$, the triangle inequality gives $\lambda(x)+\tilde{d}_x(y)\geq d_{G_0}(s,y)$. Thus $\mu(y)\geq d_{G_0}(s,y)$.

        Suppose for contradiction that $\mu(y)>d_{G_0}(s,y)$. Choose a simple shortest $s$--$y$ path, let $r$ be the last vertex of this path that belongs to $R$, and let $z$ be the vertex immediately after $r$. The vertex $z$ exists because $y$ is unsettled at the beginning of the round. Let $P_{z,y}$ be the suffix of this path from $z$ to $y$. By the choice of $r$, the suffix $P_{z,y}$ lies entirely in the current graph $G$.

        The prefix of the chosen path up to $z$ is itself a shortest path to $z$. Since $r\in R$, the inductive invariant gives $d(s,r)=d_{G_0}(s,r)$, and therefore
        \[
            d_{G_0}(s,z)=d_{G_0}(s,r)+w(r,z)=d(s,r)+w(r,z).
        \]
        Thus $\lambda^*(z)\leq d(s,r)+w(r,z)=d_{G_0}(s,z)$. The reverse inequality holds because every boundary record for $z$ is the length of an actual path from $s$ to $z$. Hence
        \[
            \lambda^*(z)=d_{G_0}(s,z)\leq d_{G_0}(s,y)<\mu(y).
        \]

        We first show that $z\in S$. If not, then either $|S|<T$ and all finite boundary vertices are in $S$, contradicting that $z$ has finite boundary value, or $|S|=T$ and every $x\in S$ satisfies $\lambda(x)\leq \lambda^*(z)<\mu(y)$. In the latter case, the $T$ self-candidates of the vertices in $S$ all have value strictly smaller than $\mu(y)$, so $y$ could not be among the $T$ vertices of minimum candidate value. Therefore $z\in S$, and $\lambda(z)=d_{G_0}(s,z)$.

        If $P_{z,y}$ uses at most $T=h_0$ edges, then $d_G^{T}(z,y)\leq w(P_{z,y})$. If $y\in Y_z$, the upper-bound guarantee of Lemma~\ref{lem:usetosearch} gives a candidate for $y$ of value at most
        \[
            \lambda(z)+d_G^{T}(z,y)
            \leq d_{G_0}(s,z)+w(P_{z,y})
            =d_{G_0}(s,y)<\mu(y),
        \]
        a contradiction to the definition of $\mu(y)$. If $y\notin Y_z$ and $|Y_z|<T$, completeness gives the same contradiction. Finally, if $y\notin Y_z$ and $|Y_z|=T$, then the ordering guarantee implies that every $a\in Y_z$ has
        \[
            \lambda(z)+\tilde{d}_z(a)
            \leq \lambda(z)+d_G^{T}(z,y)
            \leq d_{G_0}(s,y)<\mu(y).
        \]
        Thus there are $T$ candidate vertices with value smaller than $\mu(y)$, again contradicting the choice of $B$.

        It remains to consider the case where $P_{z,y}$ uses more than $T$ edges. Let $a_1,\ldots,a_T$ be the first $T$ vertices after $z$ on this suffix. Each $a_r$ is reachable from $z$ within $T$ hops in $G$, and the chosen path prefix from $s$ to $a_r$ has length at most $d_{G_0}(s,y)<\mu(y)$. If all vertices $a_1,\ldots,a_T$ belong to $Y_z$, then they give $T$ candidate vertices of value smaller than $\mu(y)$. Otherwise, choose some $a_r\notin Y_z$. If $|Y_z|<T$, completeness is contradicted. If $|Y_z|=T$, ordering implies that every vertex returned in $Y_z$ has candidate value at most $\lambda(z)+d_G^T(z,a_r)$, and this is at most the length of the chosen path prefix from $s$ to $a_r$, hence smaller than $\mu(y)$. In all cases there are $T$ candidate vertices with value smaller than $\mu(y)$, contradicting the choice of~$B$. Putting everything together, we have shown that $\mu(y)=d_{G_0}(s,y)$.
    \end{proof}

    \begin{claim}\label{claim:dijkstra-progress}
        The algorithm terminates after $O(n/T)$ rounds, and when it terminates every vertex outside $R$ is unreachable from $s$.
    \end{claim}

    \begin{proof}
        If a round starts with at least $T$ unsettled vertices of finite boundary value, then $|S|=T$. The self-candidates of the vertices in $S$ ensure that at least $T$ candidate vertices exist, so the round settles exactly $T$ vertices.

        Now consider a round in which fewer than $T$ vertices are settled. Then fewer than $T$ candidate vertices were produced. We claim that no unsettled reachable vertex remains afterward. Suppose otherwise, and let $v$ be such a vertex. Choose a simple shortest path to $v$, let $r$ be the last vertex of the old settled set $R$ on this path, and let $z$ be the vertex immediately after $r$. As in the proof of Claim~\ref{claim:settled-exact}, the suffix from $z$ to $v$ lies in the old current graph and $\lambda^*(z)<+\infty$. Since fewer than $T$ vertices were settled, Step~1 must have extracted all unsettled vertices with finite current boundary key, and hence $z\in S$.

        If the suffix from $z$ to $v$ has at most $T$ edges, then $v$ is reachable from $z$ within the search radius. Since fewer than $T$ candidates were produced in the round, the returned set $Y_z$ cannot have size $T$; otherwise its vertices alone would give $T$ candidates. Thus $|Y_z|<T$, and the completeness guarantee of the search from $z$ implies that $v\in Y_z$. Then $v$ would receive a candidate label and, because fewer than $T$ candidates were produced in total, would be included in $B$, contradicting that it remains unsettled after the round.

        If the suffix from $z$ to $v$ has more than $T$ edges, let $a_1,\ldots,a_T$ be the first $T$ vertices on this suffix. All of them are reachable from $z$ within $T$ hops. If $|Y_z|=T$, then the search from $z$ already contributes $T$ candidates. If $|Y_z|<T$, completeness implies that all vertices $a_1,\ldots,a_T$ belong to $Y_z$, and again the search from $z$ contributes $T$ candidates. Both alternatives contradict the assumption that fewer than $T$ candidates were produced. Thus no reachable unsettled vertex remains after such a round. All earlier rounds settle exactly $T$ vertices, so there are $O(n/T)$ rounds in total.
    \end{proof}

    \begin{claim}\label{claim:putting-bounds}
        The total work and span satisfy the bounds in the theorem statement.
    \end{claim}

    \begin{proof}
        We first account for the LNL hierarchy. Each initialization step uses Lemma~\ref{lem:initusingother}, and each update implementation uses Lemma~\ref{lem:updatequeries}. Passing update sequences down one level increases the transcript width and the number of lower-level update queries by only a polylogarithmic factor, and also incurs the factor $k_{j+1}/k_j=2$ when the bound is rewritten in the form $O(nf/k_{j+1})$. Since $J\leq \sqrt{\log n}$, the total blowup over the hierarchy is $n^{o(1)}$.

        Therefore the total work spent building and maintaining all non-base LNL structures is bounded by
        \[
            n^{o(1)}\sum_{j=0}^{J-1}\tilde{O}(nk_j^2\rho+m)
            =\tilde{O}(n^{1+o(1)}T^2+m^{1+o(1)}),
        \]
        because $k_j=T\cdot 2^j$, $2^J\leq n^{o(1)}$, and $\rho=n^{o(1)}$. The corresponding span is
        \[
            n^{o(1)}\sum_{j=0}^{J-1}\tilde{O}\left(\frac{n\rho}{k_j}\right)
            =\tilde{O}(n^{1+o(1)}/T).
        \]
        The direct base structure contributes $\tilde{O}(m^{1+o(1)})$ work and no larger span: each edge is deleted from the base priority queues only when one of its endpoints is deleted, and rollback work is charged to the transcript records created by the same operations.

        It remains to account for the batched Dijkstra layer outside the hierarchy. Each round performs $|S|\leq T$ top-level searches with multiplier $1$. Each such search costs $\tilde{O}(T^2)$ work and $\tilde{O}(1)$ span, so over $O(n/T)$ rounds these searches use $\tilde{O}(nT^2)$ work and $\tilde{O}(n/T)$ span. Every edge of $G_0$ is inserted into the Dijkstra priority queue once, when its tail is settled, and stale records are discarded only after being extracted. The reinsertion of vertices in $S\setminus B$ contributes $O(T)$ priority-queue operations per round, and hence $O(n)$ operations in total. Thus all boundary-queue operations and edge relaxations contribute $\tilde{O}(m+n)$ work in total and $\tilde{O}(n/T)$ span over all rounds, which is within the same bounds. Combining these estimates gives $\tilde{O}(n^{1+o(1)}T^2+m^{1+o(1)})$ work and $\tilde{O}(n^{1+o(1)}/T)$ span. Since $T\geq t/n^{o(1)}$, the theorem follows.
    \end{proof}

    Claims~\ref{claim:hierarchy-correct},~\ref{claim:settled-exact}, and~\ref{claim:dijkstra-progress} prove correctness, and Claim~\ref{claim:putting-bounds} gives the required work and span bounds.
\end{proof}

\section*{Acknowledgments}
All ideas in this paper are solely due to the human authors, and the human authors are fully responsible for the contents of this paper. After writing a full draft of the paper, the authors used ChatGPT 5.5 Pro to improve the presentation of some of the proofs, which were then edited again by the authors.

\bibliographystyle{alpha}
\bibliography{refs}
\end{document}